\documentclass[letterpaper,twocolumn,10pt]{article}
\usepackage{usenix2020}

\usepackage{amsmath,amsthm,amssymb,amsfonts,bm,bbm}
\usepackage{mathtools}
\usepackage{graphicx,subfigure}
\usepackage{enumitem, color,soul,makecell}
\usepackage{mdframed}
\usepackage{algorithm}
\usepackage{algorithmic}
\usepackage{tablefootnote}

\newtheorem{theorem}{Theorem}
\newtheorem{definition}{Definition}
\newtheorem{lemma}{Lemma}

\newtheorem{example}{Example}

\usepackage{pifont}

\usepackage{threeparttable,multirow}

\author{{\rm Xiaolan Gu}\\ University of Arizona\\{\rm xiaolang@arizona.edu}
\and {\rm Ming Li}\\ University of Arizona\\{\rm lim@arizona.edu} 
\and {\rm Li Xiong}\\ Emory University\\{\rm lxiong@emory.edu}}

\begin{document}

\date{}

\title{\Large \bf DP-BREM: Differentially-Private and Byzantine-Robust Federated Learning\\ with Client Momentum}

\maketitle

\begin{abstract}
Federated Learning (FL) allows multiple participating clients to train machine learning models collaboratively while keeping their datasets local and only exchanging the gradient or model updates with a coordinating server. Existing FL protocols are vulnerable to attacks that aim to compromise data privacy and/or model robustness. Recently proposed defenses focused on ensuring either privacy or robustness, but not both. In this paper, we focus on simultaneously achieving differential privacy (DP) and Byzantine robustness for cross-silo FL, based on the idea of learning from history. The robustness is achieved via client momentum, which averages the updates of each client over time, thus reducing the variance of the honest clients and exposing the small malicious perturbations of Byzantine clients that are undetectable in a single round but accumulate over time. In our initial solution DP-BREM, DP is achieved by adding noise to the aggregated momentum, and we account for the privacy cost from the momentum, which is different from the conventional DP-SGD that accounts for the privacy cost from the gradient. Since DP-BREM assumes a trusted server (who can obtain clients' local models or updates), we further develop the final solution called DP-BREM\textsuperscript{+}, which achieves the same DP and robustness properties as DP-BREM without a trusted server by utilizing secure aggregation techniques, where DP noise is securely and jointly generated by the clients. Both theoretical analysis and experimental results demonstrate that our proposed protocols achieve better privacy-utility tradeoff and stronger Byzantine robustness than several baseline methods, under different DP budgets and attack settings.
\end{abstract}

\section{Introduction}

Federated learning (FL) \cite{mcmahan2017communication}  is an emerging paradigm that enables multiple clients to collaboratively learn models without explicitly sharing their data. The clients upload their local model updates to a coordinating server, which then shares the global average with the clients in an iterative process. This offers a promising solution to mitigate the potential privacy leakage of sensitive information about individuals (since the data stays local with each client), such as typing history, shopping transactions, geographical locations, and medical records. However, recent works have demonstrated that FL may not always provide sufficient \emph{privacy} and \emph{robustness} guarantees.  In terms of privacy leakage, exchanging the model updates throughout the training process can still reveal sensitive information \cite{bhowmick2018protection,melis2019exploiting} and cause deep leakage such as pixel-wise accurate image recovery  \cite{zhu2019deep,yin2021see}, either to a third-party (including other participating clients) or the central server. In terms of robustness, the decentralization design of FL systems opens up the training process to be manipulated by malicious clients, aiming to either prevent the convergence of the global model (a.k.a. Byzantine attacks) \cite{fang2020local,baruch2019little,xie2020fall}, or implant a backdoor trigger into the global model to cause targeted misclassification (a.k.a. backdoor attacks)  \cite{bagdasaryan2020backdoor,wang2020attack}.

To mitigate the privacy leakage in FL, Differential Privacy (DP) \cite{dwork2006calibrating,dwork2014algorithmic} has been adopted as a rigorous privacy notion. Existing frameworks  \cite{mcmahan2018learning,geyer2017differentially,li2020differentially} applied DP in FL to provide \emph{client-level} privacy under the assumption of a trusted server: whether a client has participated in the training process cannot be inferred by a third party from the released global model. Other works in FL \cite{zheng2021federated,li2020differentially,xu2019hybridalpha,truex2019hybrid} focused on \emph{record-level} privacy: whether a data record at a client has participated during training cannot be inferred by the server or other adversaries that have access to the model updates or the global model. Record-level privacy is more relevant in cross-silo (as versus cross-device) FL scenarios, such as multiple hospitals collaboratively learn a prediction model for COVID-19, in which case what needs to be protected is the privacy of each patient (corresponding to each record in a hospital's dataset). In this paper, we focus on cross-silo FL with \emph{record-level} DP, where each client possesses a set of raw records, and each record corresponds to an individual's private data.

To defend against Byzantine attacks, robust FL protocols are proposed to ensure that the training procedure is robust to a fraction of potentially malicious clients. This problem has received significant attention from the community. Most existing approaches replace the averaging step at the server with a robust aggregation rule, such as the median \cite{chen2017distributed,yin2018byzantine,blanchard2017machine,mhamdi2018hidden}. However, recent state-of-the-art attacks \cite{baruch2019little,xie2020fall,shejwalkar2021manipulating} have demonstrated the failure of the above robust aggregators. Furthermore, a recent work \cite{karimireddy2021learning} shows that there exist realistic scenarios where these robust aggregators fail to converge, even if there are no Byzantine attackers and the data distribution is identical (i.i.d.) across the clients, and proposed a new solution called Learning From History (LFH) to address this issue.  LFH achieves robustness via client momentum with the motivation of averaging the updates of each client over time, thus reducing the variance of the honest clients and exposing the small malicious perturbations of Byzantine clients that are undetectable in a single round but accumulate over time.

In this paper, we focus on achieving record-level DP and Byzantine robustness simultaneously in cross-silo FL. Existing  FL protocols with DP-SGD \cite{abadi2016deep} do not achieve the robustness property intrinsically.  Directly implementing an existing robust aggregator over the privatized client gradients will lead to poor utility because these aggregators (such as median \cite{yin2018byzantine,blanchard2017machine,mhamdi2018hidden}) usually have large sensitivity and require large DP noise, leading to poor utility. It is desirable to achieve DP guarantees based on average-based aggregators. Although LFH \cite{karimireddy2021learning} is an average-based Byzantine-robust FL protocol, it aggregates client momentum instead of gradient, thus it is non-trivial to achieve DP on top of LFH. We show that a direct combination of LFH with DP-SGD momentum has several limitations, leading to both poor utility and robustness. Therefore, we aim to address these limitations in our solution.

To achieve an enhanced privacy-utility tradeoff, we start our problem from an assumption that the server is trusted and developed a \textbf{D}ifferentially-\textbf{P}rivate and \textbf{B}yzantine-\textbf{R}obust f\textbf{E}derated learning algorithm with client \textbf{M}omentum (\textbf{DP-BREM}), which essentially is a DP version of the Byzantine-robust method LFH \cite{karimireddy2021learning}. Instead of adding DP noise to the gradient and then aggregating momentum as post-processing, we add DP noise to the aggregated momentum with carefully computed sensitivity to account for the privacy cost. Since the noise is added to the final aggregate (instead of intermediate local gradient), our basic solution DP-BREM maintains the non-private LFH's robustness as much as possible, which we show both theoretically (via convergence analysis) and empirically (via experimental results). Then, we relax our trust assumption to a malicious server (for privacy only) and develop our final solution DP-BREM\textsuperscript{+}. It utilizes secure multiparty computation (MPC) techniques, including secure aggregation and secure noise generation, to achieve the same DP and robustness guarantees as in DP-BREM. In Table \ref{tab:DP_compare}, we compare DP-BREM and DP-BREM\textsuperscript{+} with existing approaches (or the variants) that achieve both DP and Byzantine robustness (DDP-RP \cite{wang2022privacy} and DP-RSA \cite{zhu2022bridging} are described in Sec. 7). These approaches will be evaluated and compared in experiments. Our main contributions are:

1) We propose DP-BREM, a novel differentially private and Byzantine-robust FL protocol that adds DP noise to aggregated client momentum with computed sensitivity. Our privacy analysis (Theorem \ref{thm:privacy_analysis}) accounts for momentum, differing from conventional DP-SGD which accounts for gradient. Our convergence analysis (Theorem \ref{thm:convergence_rate}) shows minimal convergence rate sacrifice for DP compared to the baseline.

2) Considering that DP-BREM is developed under the assumption of a trusted server, we propose the final solution called DP-BREM\textsuperscript{+} (Section \ref{sec:secure_aggregation}), which achieves the same privacy and robustness properties as DP-BREM, even under a \emph{malicious} server (for privacy only), using secure multiparty computation techniques. DP-BREM\textsuperscript{+} is built based on the framework of secure aggregation with verifiable inputs (SAVI) \cite{roy2022eiffel}, but extends it to guarantee the integrity of DP noise via a novel secure distributed noise generation protocol. Our extended SAVI protocol is general enough to be applied to other DP and robust FL protocols that are average-based.

3) We demonstrate our protocols' effectiveness through extensive experiments on MNIST, CIFAR-10, and FEMNIST datasets (Section \ref{sec:experiments}), showing improved utility with the same record-level DP guarantees and strong robustness against Byzantine clients under state-of-the-art attacks, compared to baseline methods.

\section{Preliminaries}
\subsection{Differential Privacy (DP)}
\label{sec:preliminaries_DP}

Differential Privacy (DP) is a rigorous mathematical framework for the release of information derived from private data. Applied to machine learning, a differentially private training mechanism allows the public release of model parameters with a strong privacy guarantee: adversaries are limited in what they can learn about the original training data based on analyzing the parameters, even when they have access to arbitrary side information. The formal definition is as follows:
\begin{definition}[$(\epsilon,\delta)$-DP \cite{dwork2014algorithmic,dwork2006calibrating}]
\label{def:DP}
    For $\epsilon\in[0,\infty)$ and $\delta\in[0,1)$, a randomized mechanism $\mathcal{M}:\mathcal{D}\rightarrow \mathcal{R}$ with a domain $\mathcal{D}$ (e.g., possible training datasets) and range $\mathcal{R}$ (e.g., all possible trained models) satisfies $(\epsilon,\delta)$-Differential Privacy (DP) if for any two neighboring datasets $D, D^\prime\in\mathcal{D}$ that differ in only one record and for any subset of outputs $S\subseteq \mathcal{R}$, it holds that
    \begin{align*}
        \mathbb{P}[\mathcal{M}(D)\in S]\leqslant e^\epsilon\cdot \mathbb{P}[\mathcal{M}(D^\prime)\in S] +\delta
    \end{align*}
    where $\epsilon$ and $\delta$ are privacy parameters (or privacy budget), and a smaller $\epsilon$ and $\delta$ indicate a more private mechanism.
\end{definition}

\textbf{Gaussian Mechanism.} 
A common paradigm for approximating a deterministic real-valued function $f:\mathcal{D}\rightarrow\mathbb{R}$ with a differentially-private mechanism is via additive noise calibrated to $f$'s sensitivity $s_f$, which is defined as the maximum of the absolute distance $|f(D)-f(D^\prime)|$. The Gaussian Mechanism is defined by $\mathcal{M}(D)=f(D)+\mathcal{N}(0,s_f^2\cdot\sigma^2)$, where $\mathcal{N}(0,s_f^2\cdot\sigma^2)$ is noise drawn from a Gaussian distribution. It was shown that $\mathcal{M}$ satisfies $(\epsilon,\delta)$-DP if $\delta\geqslant\frac{4}{5}e^{-(\sigma\epsilon)^2/2}$ and $\epsilon<1$ \cite{dwork2014algorithmic}. Note that we use an advanced privacy analysis tool proposed in \cite{dong2019gaussian}, which works for all $\epsilon>0$. 

\textbf{DP-SGD Algorithm.}
The most well-known differentially-private algorithm in machine learning is DP-SGD \cite{abadi2016deep}, which introduces two modifications to the vanilla stochastic gradient descent (SGD). First, a \emph{clipping step} is applied to the gradient so that the gradient is in effect bounded for a finite sensitivity. The second modification is \emph{Gaussian noise augmentation} on the summation of clipped gradients, which is equivalent to applying the Gaussian mechanism to the updated iterates. The privacy accountant of DP-SGD is shown in Appendix \ref{apx:GDP}.

\subsection{Federated Learning (FL) with DP}
\label{sec:preliminaries_FL_DP}

Federated Learning (FL) \cite{kairouz2021advances,mcmahan2017communication} is a collaborative learning setting to train machine learning models. We consider the horizontal cross-silo FL setting, which involves multiple clients, each holding their own private dataset of the same set of features, and a central server that implements the aggregation. Unlike the traditional centralized approach, data is not stored at a central server; instead, clients train models locally and exchange updated parameters with the server, which aggregates the received local model parameters and sends them to the clients. Based on the participating clients and scale, federated learning can be classified into two types: \emph{cross-device} FL where clients are typically mobile devices and the client number can reach up to a scale of millions; \emph{cross-silo} FL (our focus) where clients are organizations or companies and the client number is relatively small (e.g., within hundreds).


\textbf{FL with DP.} In FL, the \emph{neighboring datasets} $D$ and $D^\prime$ in Definition \ref{def:DP} can be defined at two distinct levels:  \emph{record-level} and \emph{client-level}. In cross-device FL, each device usually stores one individual's data, and then the whole device's data should be protected. It corresponds to client-level DP, where $D^\prime$ is obtained by adding or removing one client/device's whole training dataset from $D$. In cross-silo FL, each record corresponds to one individual's data, then record-level DP should be provided, where $D^\prime$ is obtained by adding or removing a single training record/example from $D$. Since we consider cross-silo FL, achieving \emph{record-level} DP is our privacy goal.

\subsection{Byzantine Attacks and Defenses}
\label{sec:preliminaries_LFH}

In a Byzantine attack, the adversary aims to destroy the convergence of the model. Due to the decentralization design, FL systems are vulnerable to Byzantine clients, who may not follow the protocol and can send arbitrary updates to the server. Also, they may have complete knowledge of the system and can collude with each other. Most state-of-the-art defense mechanisms \cite{chen2017distributed,yin2018byzantine,blanchard2017machine,mhamdi2018hidden} play with median statistics of client gradients. However, recent attacks \cite{baruch2019little,xie2020fall} have empirically demonstrated the failure of the above robust aggregations. 

\textbf{LFH: Non-private Byzantine-Robust Defense.}
Recently, Karimireddy et al. \cite{karimireddy2021learning} showed that most state-of-the-art robust aggregators require strong assumptions and may not converge even in the complete absence of Byzantine attackers. Then, they proposed a new Byzantine-robust scheme called "learning from history" (LFH) that essentially utilizes two simple strategies: \emph{client momentum} (during local update) and \emph{centered clipping} (during server aggregation). In each iteration $t$,  client $\mathsf{C}_i$ receives the global model parameter $\bm{\theta}_{t-1}$ from the server, and computes the local gradient of the random dataset batch  $\mathcal{D}_{i,t}\subseteq\mathcal{D}_{i}$ by
\begin{align}
    \label{equ:LFH_local_gradident}
    \bm{g}_{t,i}=\frac{1}{|\mathcal{D}_{i,t}|}\sum\nolimits_{\bm{x}\in\mathcal{D}_{i,t}}\nabla_{\bm{\theta}}\ell(\bm{x}, \bm{\theta}_{t-1})
\end{align}
where $\nabla_{\bm{\theta}}\ell(\bm{x}, \bm{\theta}_{t-1})$ is the per-record gradient w.r.t. the loss function $\ell(\cdot)$. The client momentum can be computed via
\begin{align}
    \label{equ:LFH_local_momentum}
    \bm{m}_{t,i} =
    (1-\beta)\bm{g}_{t,i} + \beta\bm{m}_{t-1,i}
\end{align}
where $\beta\in[0,1)$. After receiving $\bm{m}_{t, i}$ from all clients, the server implements aggregation with centered clipping via
\begin{align}
    \label{equ:LFH_server_aggregate}
    \bm{m}_t = \bm{m}_{t-1} + \frac{1}{n}\sum\nolimits_{i=1}^n \mathsf{Clip}_C(\bm{m}_{t,i} - \bm{m}_{t-1}) 
\end{align}
where $\mathsf{Clip}_C(\cdot)$ with scalar $C>0$ is the clipping function:
\begin{align}
    \label{equ:clipping_function}
    \mathsf{Clip}_C(\bm{x})\coloneqq \bm{x}\cdot\min\{1,~C/\|\bm{x}\|\}
\end{align}
and $\|\bm{x}\|$ is the L2-norm of any vector $\bm{x}$. The clipping operation $\mathsf{Clip}_C(\bm{m}_{t,i} - \bm{m}_{t-1})$ essentially bounds the distance between client's local momentum $\bm{m}_{t,i}$ and the previous aggregated momentum $\bm{m}_{t-1}$, thus restricts the impact from Byzantine clients. Then, the global model $\bm{\theta}_t$ can be updated by $\bm{\theta}_t = \bm{\theta}_{t-1}-\eta_t \bm{m}_t$ with learning rate $\eta_t$. The convergence rate under Byzantine attacks is shown by the following lemma.

\begin{lemma}[Convergence Rate of LFH \cite{karimireddy2021learning}]
\label{lem:convergence_rate_of_LFH}
With some parameter tuning, the convergence rate of the Byzantine-robust algorithm LFH is asymptotically (ignoring constants and higher order terms) of the order
\begin{align}
    \label{equ:convergence_rate_of_LFH}
    \frac{1}{T}\sum\nolimits_{t=1}^T\mathbb{E}\|\nabla\ell(\bm{\theta}_{t-1})\|^2 \lesssim \sqrt{\frac{\rho^2}{T}\frac{1+|\mathcal{B}|}{n}}
\end{align}
where $\ell(\cdot)$ is the loss function, $T$ is the total number of training iterations, $|\mathcal{B}|$ is the number of Byzantine clients, $n$ is the number of all clients, and $\rho$ is a parameter that quantifies the variance of honest clients' stochastic gradients:
\begin{align}
    \label{equ:definition_rho_LFH}
    \mathbb{E}\|\bm{g}_{t,i}-\mathbb{E}[\bm{g}_{t,i}]\|^2\leqslant\rho^2
\end{align}
\end{lemma}

\textbf{Interpretation of Lemma \ref{lem:convergence_rate_of_LFH}.} 
When there are no Byzantine clients, LFH recovers the optimal rate of $\frac{\rho}{\sqrt{nT}}$. In the presence of a $|\mathcal{B}|/n$ fraction of Byzantine clients, the rate has an additional term $\rho\sqrt{\frac{|\mathcal{B}|/n}{T}}$, which depends on the fraction $|\mathcal{B}|/n$ but does not improve with increasing clients.

\section{Problem Statement and Motivation}

\subsection{Problem Statement}
\label{sec:problem_statement}

\textbf{System Model.}
Our system model follows the general setting of Fed-SGD \cite{mcmahan2017communication}. There are multiple parties in the FL system: one aggregation server and $n$ participating clients $\{\mathsf{C}_1,\cdots,\mathsf{C}_n\}$.  The server holds a global model $\bm{\theta}_t\in\mathbb{R}^d$  and each client $\mathsf{C}_i$, $i\in\{1,\cdots,n\}$ possesses a private training dataset $\mathcal{D}_i$. The server communicates with each client through a secure (private and authenticated) channel. During the iterative training process, the server broadcasts the global model in the current iteration to all clients and aggregates the received gradient/momentum from all clients (or a subset of clients) to update the global model until convergence. The final global model is returned after the training process as the output.

\textbf{Threat Model.}
The considered adversary aims to perform a 1) privacy attack and/or 2) Byzantine attack with the following threat model, respectively.

1) Privacy Attack. Following the conventional FL setting, we assume the server has no access to the client's local training data, but may have an incentive to infer clients' private information. In our initial solution called DP-BREM, we assume a trusted server that can obtain clients' local models/updates. The adversary is a third party or the participating clients (can be any set of clients) who have access to the intermediate and final global models and may use them to infer the private data of an honest client $\mathsf{C}_i$. Hence, the privacy goal is to ensure the global model (and its update) satisfies DP. In our final solution DP-BREM\textsuperscript{+}, in addition to third parties and clients, the adversary also includes the server that tries to infer additional information from the local updates (and may deviate from the protocol for privacy inference). Such a model is also adopted in previous work \cite{roy2022eiffel}. Following \cite{roy2022eiffel}, we assume a minority of malicious clients who can deviate from the protocol arbitrarily. 

2) Byzantine Attack. 
Recall that the goal of Byzantine attacks is to destroy the convergence of the global model (discussed in Section \ref{sec:preliminaries_LFH}). We only consider malicious clients as the adversaries for Byzantine attacks because the server's primary goal is to train a robust model, thus no incentive to implement Byzantine attacks. These malicious clients (assumed to be a minority of all participating clients) can deviate from the protocol arbitrarily and have full control of both their local training data and their submission to the servers, but do not influence other honest clients.

\textbf{Objectives.}
The goal of this paper is to achieve both record-level DP and Byzantine robustness at the same time. We aim to provide high utility (i.e., high accuracy of the global model) with the required DP guarantee under the existence of Byzantine attacks from malicious clients.  Our ultimate privacy goal is to provide DP guarantees against an untrusted server and other clients, but we start by assuming a trusted server first in our initial solution.

\begin{table*}[t!]
	\footnotesize
	\centering
	\caption{Comparison of FL approaches with DP and Byzantine-robustness}
	\renewcommand{\arraystretch}{1.8}
	\begin{threeparttable}
	\begin{tabular}{c|cc|cc|c}
		\Xhline{1pt}
		\multirow{2}{6em}{\makecell[c]{Approaches}} & 
		 \multicolumn{4}{c|}{ Differential Privacy (DP) \tnote{$\mathsection$}} & \multicolumn{1}{c}{Byzantine Robustness}  \\
		 \cline{2-6}
		 & \makecell[c]{ Trust Assumption \\ of Server}  &
		\makecell[c]{Noise \\ Generator } &
		\makecell[c]{Perturbation \\ Mechanism} &
		 \makecell[c]{Standard Deviation    \\ of Noise in Aggregate} &  \makecell[c]{Mechanism\\   }   \\
		\Xhline{1pt}
		\makecell[c]{DP-FedSGD \cite{mcmahan2018learning}    \\ with both record and \\ client norm clippings} &  trusted & server & $\sum_i g_i+\mathcal{N}(0,\sigma^2)$ &   $\sigma$  & client norm clipping \\
           \hline
           \makecell[c]{CM \cite{yin2018byzantine} with \\ DP noise} &  trusted & server & $ \text{median}( \{g_i\}_{i=1}^n)+\mathcal{N}(0,\sigma^2)$ &   $\sigma$  & coordinate-wise median (CM) \\
           \hline
		DDP-RP \cite{wang2022privacy} \tnote{$\diamond$}   & honest-but-curious  & \makecell[c]{clients \\ (distributively)} & $\sum_i[g_i+\mathcal{N}(0,\frac{\sigma^2}{\tau})]$ & $\sqrt{\frac{n}{\tau}}\cdot\sigma$ &  element-wise range proof  \\
		\hline
		DP-RSA \cite{zhu2022bridging} & untrusted  & client& $\sum_i[\text{sign}(g_i)+\mathcal{N}(0,\sigma^2)]$ &   $\sqrt{n}\cdot\sigma$  & aggregation of sign-SGD \\
            \hline
		\makecell[c]{DP-LFH \\ (baseline in Sec. \ref{sec:challenges_and_baseline})} & untrusted   & client & $\sum_i[m_i+\mathcal{N}(0,\sigma^2)]$ &   $\sqrt{n}\cdot\sigma$  & \multirow{3}{*}{\makecell[c]{LFH \cite{karimireddy2021learning}: client momentum \\ and centered clipping}} \\
		\cline{1-5}
		\makecell[c]{DP-BREM \\(our initial solution)}  & trusted  & server & \multirow{2}{*}{$\sum_i m_i+\mathcal{N}(0,\sigma^2)$} & \multirow{2}{*}{$\sigma$} &   \\
        \cline{1-3}
		\makecell[c]{DP-BREM\textsuperscript{+} \\ (our final solution)  \tnote{$\dagger$}} & untrusted  & \makecell[c]{clients \\ (jointly)} &  &  &  \\
		\Xhline{1pt}
	\end{tabular}
	\begin{tablenotes}
    \item[$\mathsection$] We demonstrate the privacy-utility tradeoff by comparing the standard deviation of DP noise on the aggregation, with smaller values indicating less negative impact on utility. Note that different approaches use different aggregation strategies, where $g_i$ is the local gradient and $m_i$ is the local momentum.
     \item[$\diamond$]  DDP-RP assumes an honest-but-curious server and ensures distributed DP (DDP) with secure aggregation. Clients add partial noise with a smaller standard deviation based on the number of honest clients, $\tau$, resulting in a better privacy-utility tradeoff than local DP (LDP).
     \item[$\dagger$] DP-BREM\textsuperscript{+} matches DP-BREM's DP and robustness guarantees with a different server trust assumption. It achieves central DP without a trusted server, as clients securely generate and add noise using the proposed noise generation and secure aggregation protocols.
   \end{tablenotes}
   \end{threeparttable}
	\label{tab:DP_compare}
	\vspace{-5mm}
\end{table*}

\subsection{Challenges and Baseline}
\label{sec:challenges_and_baseline}

\textbf{Challenges.}
Replacing the average-based aggregator with median-based or complex robust aggregators increases DP sensitivity. Achieving both DP and Byzantine robustness with high utility is challenging because these methods result in significantly larger DP sensitivity than averaging, as illustrated in Example \ref{exp:sensitivity_average_median}.

\begin{example}[Sensitivity Computation: Average vs. Median]
\label{exp:sensitivity_average_median}
Consider a dataset with 5 samples: $\mathcal{D}=\{1,3,5,7,9\}$, and its neighboring dataset $\mathcal{D}^\prime$ is obtained by changing one value in $\mathcal{D}$ with at most 1, such as $\mathcal{D}^\prime=\{1,3,\bm{4},7,9\}$. Then, the sensitivity of average-query is $\max\limits_{\mathcal{D}, \mathcal{D}^\prime}|\mathsf{avg}(\mathcal{D})-\mathsf{avg}(\mathcal{D}^\prime)|=1/5=0.2$. However, the sensitivity of median-query is $\max\limits_{\mathcal{D}, \mathcal{D}^\prime}|\mathsf{median}(\mathcal{D})-\mathsf{median}(\mathcal{D}^\prime)|=1$. Moreover, when increasing the size of the dataset, the sensitivity of the average query will be reduced (and then less noise to be added), while the sensitivity of the median query is the same.
\end{example}


\textbf{DP-LFH: baseline via direct combination of LFH and DP-SGD.}
As shown in Section \ref{sec:preliminaries_LFH}, the Byzantine-robust scheme LFH \cite{karimireddy2021learning} utilizes an average-based aggregator, which can be regarded as a non-private robust solution to address the disadvantage of the median-based aggregator. A straightforward method to add DP protection on top of LFH is to combine it with the DP-SGD algorithm. However,  LFH requires each client to compute the local momentum $\bm{m}_{t, i}$ for server aggregation, while DP-SGD aggregates gradients and accounts for the privacy cost via the composition of iterative gradient update. In LFH, since the gradient is computed only on the client-side, a straightforward solution to integrate DP is to use DP-SGD at each client to privatize the local gradient, and then compute the momentum from the privatized gradient (thus there is no additional privacy cost due to post-processing).  Formally,  client $\mathsf{C}_i$ computes

\begin{align}
    \label{equ:baseline_local_gradient}
    \bm{g}_{t,i}=\frac{1}{|\mathcal{D}_{i,t}|}\sum_{\bm{x}\in\mathcal{D}_{i,t}}\mathsf{Clip}_R(\nabla_{\bm{\theta}}\ell(\bm{x}, \bm{\theta}_{t-1}))+\mathcal{N}(0, R^2\sigma^2\mathbf{I}_d),
\end{align}
where $\mathbf{I}_d$ is an identity matrix with size $d\times d$ ($d$ is the model size, i.e., $\bm{\theta}_t\in\mathbb{R}^d$), the record-level clipping $\mathsf{Clip}_R(\cdot)$ restricts the sensitivity when adding/removing one record from the local dataset, and Gaussian noise $\mathcal{N}(0, R^2\sigma^2\mathbf{I}_d)$ introduces DP property on $\bm{g}_{t, i}$. Since DP is immune to post-processing, the remaining steps can follow the original LFH without additional privacy costs. This baseline solution DP-LFH achieves record-level DP against an untrusted server but has limitations, leading to poor privacy-utility tradeoff and robustness.

\textbf{Limitation 1: large aggregated noise.}  
Since each client locally adds DP noise, the overall noise after aggregation is larger than the case of the central setting under the same privacy budget $\epsilon$ since only the server adds noise in the central setting. Therefore, DP-LFH has a poor privacy-utility tradeoff.

\textbf{Limitation 2: large impact on Byzantine robustness.} 
Since the DP noise is added locally to each client's gradient before momentum-based clipping, it leads to a negative impact on Byzantine robustness: the noisy client momentum $\bm{m}_{t, i}$ has larger variance than the noise-free one, which leads to larger bias and variance on the clipping step $\mathsf{Clip}_C(\bm{m}_{t, i} - \bm{m}_{t-1})$. Furthermore, this impact will be enlarged when there are more Byzantine clients, which is explained as follows. Since the parameter $\rho^2$ defined in \eqref{equ:definition_rho_LFH} quantifies the variance of client's gradient, and the DP noise is added to the local gradient in \eqref{equ:baseline_local_gradient}, the parameter $\rho$ of the convergence rate shown in \eqref{equ:convergence_rate_of_LFH} is replaced by $\rho+\sqrt{d}\sigma$ (ignoring constants) for DP-LFH, i.e., the convergence rate of DP-LFH is asymptotic of the order
\begin{align}
    \label{equ:convergence_rate_of_baseline}
    \frac{1}{T}\sum\nolimits_{t=1}^T\mathbb{E}\|\nabla\ell(\bm{\theta}_{t-1})\|^2 \lesssim \sqrt{\frac{(\rho+\sqrt{d}\sigma)^2}{T}\frac{1+|\mathcal{B}|}{n}}
\end{align}
Therefore, either a large $d$ (i.e., large model) or a large $\sigma$ (i.e., small privacy budget $\epsilon$) will enlarge the impact from Byzantine clients due to the order $O(\sqrt{d\sigma^2|\mathcal{B}|})$ of convergence rate. We note that Guerraoui et al.'s work \cite{guerraoui2021differential} also shares a similar insight: they show that DP with local noise and Byzantine robustness are incompatible, especially when the dimension of model parameters $d$ is large.

\textbf{Limitation 3: no privacy amplification from client-level sampling due to momentum.}
According to the recursive representation $\bm{m}_{t, i} =(1-\beta)\bm{g}_{t, i} + \beta\bm{m}_{t-1, i}$, client $\mathsf{C}_i$'s momentum in $t$-th iteration $\bm{m}_{t, i}$  is essentially a weighted summation of all previous privatized client gradients: 
\begin{align}
    \label{equ:momentum_gradient}
    \bm{m}_{t,i}
    =(1-\beta)(\bm{g}_{t,i}+\beta \bm{g}_{t-1, i}+\cdots+\beta^{t-2}\bm{g}_{2,i})+\beta^{t-1} \bm{g}_{1,i}
\end{align}
where $\bm{g}_{1,i}, \bm{g}_{2,i}, \cdots, \bm{g}_{t,i}$ are already privatized via local noise. Assume the server samples a subset of clients for aggregation in each iteration. If client $\mathsf{C}_i$'s momentum $\bm{m}_{t,i}$ is not selected in the $t$-th iteration, the aggregate is independent of $\bm{g}_{t,i}$. However, in a later iteration (i.e., $\tau>t$), if client $\mathsf{C}_i$'s momentum $\bm{m}_{\tau,i}$ is included,  it depends on $\bm{g}_{t,i}$ according to \eqref{equ:momentum_gradient}. Thus, we must account for the privacy cost of $\bm{g}_{t, i}$ in all iterations. Sampling clients offers no privacy amplification, resulting in high privacy costs or low utility.

\section{DP-BREM}
\label{sec:our_solution}
To address DP-LFH's limitations, we propose DP-BREM, a differentially-private LFH variant assuming a trusted server that generates DP noise. DP-BREM maintains robustness of LFH and uses a different privacy accountant (Theorem \ref{thm:privacy_analysis}) than DP-SGD. We also provide convergence analysis (Theorem \ref{thm:convergence_rate}) showing minimal deviation from LFH. We further relax the server trust assumption in DP-BREM\textsuperscript{+} (Section \ref{sec:secure_aggregation}) by using secure multiparty computation for secure aggregation and joint noise generation, achieving the same DP and robustness guarantees without a trusted server.


\subsection{Algorithm Design}
\label{sec:our_algorithm}

The illustration of our design is shown in Figure \ref{fig:framework}, and the algorithm is shown in Algorithm \ref{alg:DP_BREM}, where all clients need to implement local updates (in Line-3), but only a subset of their momentum vectors are aggregated by the server (in Line-4). The details of client updates and server aggregation are described below.

\textbf{Client Updates.} 
The client $\mathsf{C}_i$ first samples a random batch $\mathcal{D}_{i,t}$ from the local dataset $\mathcal{D}_{i}$ with sampling rate $p_{i}$, clips the per-record gradient $\nabla_{\bm{\theta}}\ell(\bm{x}, \bm{\theta}_{t-1})$ by $R$ and multiplies the sum by a constant factor $\frac{1}{p_{i}|\mathcal{D}_{i}|}$ to get the averaged gradient
\begin{equation}
    \label{equ:local_gradient}
    \bar{\bm{g}}_{t,i}=\frac{1}{p_{i}|\mathcal{D}_{i}|}\sum\nolimits_{\bm{x}\in\mathcal{D}_{i,t}}\mathsf{Clip}_R(\nabla_{\bm{\theta}}\ell(\bm{x}, \bm{\theta}_{t-1}) )
\end{equation}
where $\mathsf{Clip}_R(\cdot)$ is the clipping function defined in \eqref{equ:clipping_function}, but is used here to bound the sensitivity for DP (refer to DP-SGD discussed in Section \ref{sec:preliminaries_DP}). $\mathcal{D}_{i,t}$ represents a random subset obtained via subsampling from client $\mathsf{C}_i$’s dataset. This subsampling is essential to apply privacy amplification, enabling the privacy accountant to derive a tight privacy budget $\epsilon$. Note that the batch size $|\mathcal{D}_{i,t}|$ is random and $\mathbb{E}[|\mathcal{D}_{i,t}|]=p_{i}|\mathcal{D}_{i}|$. Then, the local momentum can be computed by
\begin{align}
    \label{equ:local_momentum}
    \bar{\bm{m}}_{t,i} =
    \begin{cases}
    \bar{\bm{g}}_{t,i}, & \text{if } t=1\\
    (1-\beta)\bar{\bm{g}}_{t,i} + \beta\bar{\bm{m}}_{t-1,i}, & \text{if } t>1 \\
    \end{cases}
\end{align}
where $\beta\in[0,1)$ is the momentum parameter.

\begin{figure}[!t]
    \centering
    \includegraphics[width=3.3in]{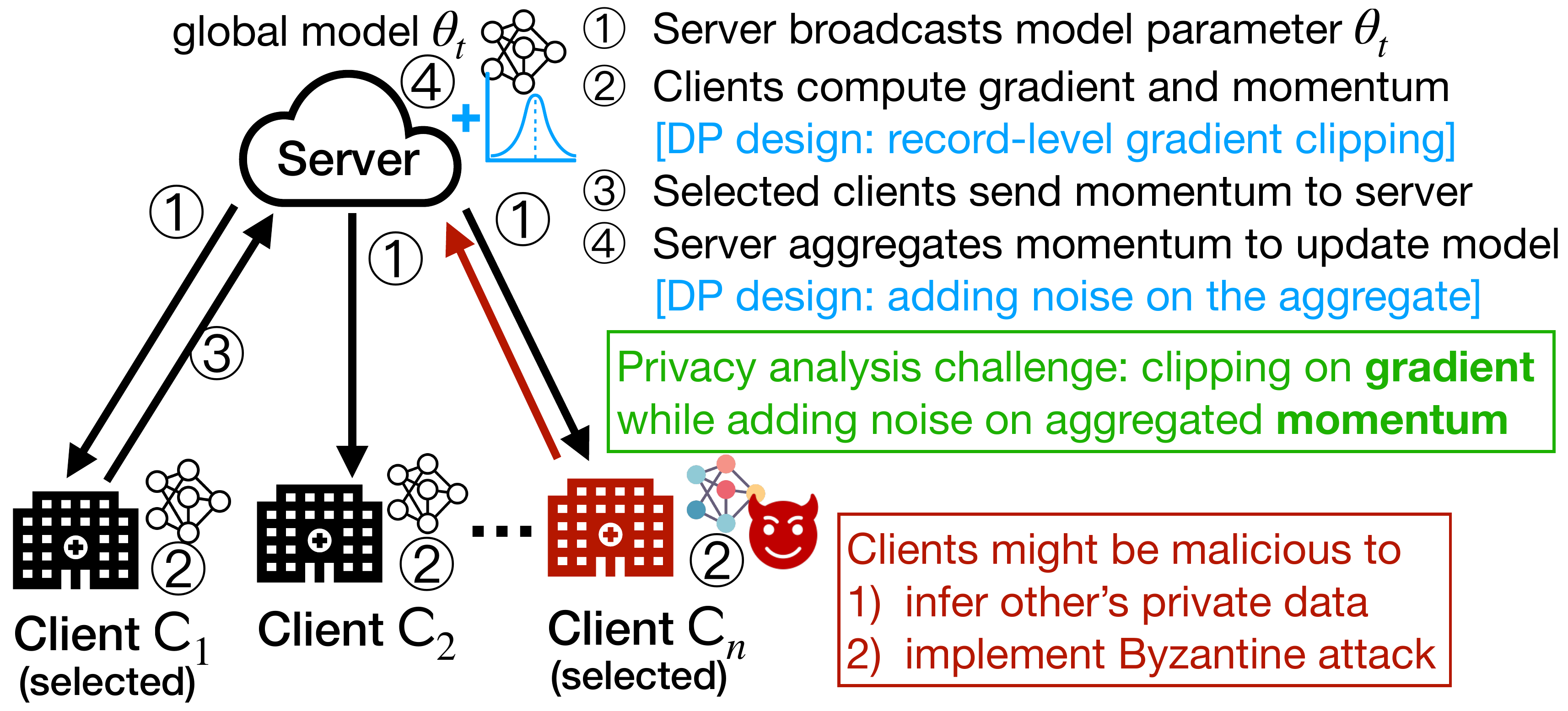}
    \vspace{-4mm}
    \caption{Illustration of our DP-BREM algorithm.}
    \vspace{-3mm}
    \label{fig:framework}
\end{figure}

\textbf{Server Aggregation.} 
The server implements centered clipping with clipping parameter $C>0$ to bound the difference between client momentum $\bar{\bm{m}}_{t,i}$ and the previous noisy global momentum $\tilde{\bm{m}}_{t-1}$ for robustness. Then, it adds Gaussian noise with standard deviation $R\sigma$ (thus the variance is $R^2\sigma^2$) to the sum of clipped terms to get the noisy global momentum $\tilde{\bm{m}}_t$ 
\begin{align}
    \label{equ:server_aggregate}
    \tilde{\bm{m}}_t = \tilde{\bm{m}}_{t-1} + \frac{1}{|\mathcal{I}_t|}\left[\sum\nolimits_{i\in\mathcal{I}_t}\mathsf{Clip}_C(\bar{\bm{m}}_{t,i} - \tilde{\bm{m}}_{t-1}) + \mathcal{N}(0,R^2\sigma^2 \mathbf{I}_d)\right] 
\end{align}
where $\mathbf{I}_d$ is an identity matrix with size $d\times d$, and only the sampled clients in $\mathcal{I}_t$ (which is obtained in Line-2 of Algorithm \ref{alg:DP_BREM} with sampling rate $q$) are aggregated in $t$-th iteration. Note that adding noise $\mathcal{N}(0,R^2\sigma^2 \mathbf{I}_d)$ to the \emph{summation} of clipped client momentum $\sum\nolimits_{i\in\mathcal{I}_t}\mathsf{Clip}_C(\bar{\bm{m}}_{t,i} - \tilde{\bm{m}}_{t-1})$ is equivalent to adding noise $\frac{1}{|\mathcal{I}_t|}\mathcal{N}(0,R^2\sigma^2 \mathbf{I}_d)$ to the \emph{average} result $\frac{1}{|\mathcal{I}_t|}\sum\nolimits_{i\in\mathcal{I}_t}\mathsf{Clip}_C(\bar{\bm{m}}_{t,i} - \tilde{\bm{m}}_{t-1})$. Then, the server updates the global model $\bm{\theta}_t$ with learning rate $\eta_t$
\begin{align}
    \label{equ:model_update}
    \bm{\theta}_t = \bm{\theta}_{t-1}-\eta_t \tilde{\bm{m}}_t
\end{align}

\begin{algorithm}[!t]
	\caption{DP-BREM}
        \small
	\begin{algorithmic}[1]
		\REQUIRE Initialization $\bm{\theta}_0\in\mathbb{R}^d$,  clipping bounds $R$ and $C$, learning rate $\eta_t$ of the global model, total number of iterations $T$, client-level sampling rate $q$, record-level sampling rate $p_i$.
		\FOR{$t=1,\cdots,T$}
		\STATE The server broadcasts the previous model $\bm{\theta}_{t-1}$ to all clients $\{\mathsf{C}_i\}_{i=1}^n$ and selects a subset of client index $\mathcal{I}_t\subseteq\{1,\cdots,n\}$, where each client is selected with probability $q$.
		\STATE Each client $\mathsf{C}_i$ for $i\in\{1,\cdots,n\}$ implements the local updates via \eqref{equ:local_gradient} and \eqref{equ:local_momentum}, while only selected clients need to send the local momentum $\bm{m}_{t,i}$ (for $i\in\mathcal{I}_t$) to the server.
		\STATE The server aggregates received clients' momentum (only for $i\in\mathcal{I}_t$)  with centered clipping and DP noise via \eqref{equ:server_aggregate}, then updates the global model $\bm{\theta}_t$ via \eqref{equ:model_update}.
		\ENDFOR
		\ENSURE The final model parameter $\bm{\theta}_T$.
	\end{algorithmic}
	\label{alg:DP_BREM}
\end{algorithm}

\textbf{Remark: clipping bounds and sampling rates.}
In our algorithm, we use two clipping bounds and two sampling rates. For clipping bounds, each client uses record-level bound $R$ to bound the \emph{per-record} gradient in \eqref{equ:local_gradient} for a finite sensitivity in record-level DP; while the server uses client-level bound $C$ to bound the difference between \emph{client momentum} $\bar{\bm{m}}_{t, i}$ and the previous noisy global momentum $\tilde{\bm{m}}_{t-1}$ in \eqref{equ:server_aggregate}, which achieves Byzantine robustness as in LFH. For sampling rates, the client $\mathsf{C}_i$ samples a batch of records $\mathcal{D}_{i,t}$ from the local dataset $\mathcal{D}_{i}$ with sampling rate $p_i$, which provides privacy amplification for DP from \emph{record-level} sampling; while the server samples a subset of clients with sampling rate $q$ (where the sampled clients set is denoted by $\mathcal{I}_t$), which provides privacy amplification for DP from \emph{client-level} sampling.

\textbf{Remark: comparison with  non-private LFH.} Comparing with the original non-private Byzantine-robust method LFH \cite{karimireddy2021learning} (see Section \ref{sec:preliminaries_LFH}), our differentially-private version has three differences. First, comparing with \eqref{equ:LFH_local_gradident}, the client gradient in \eqref{equ:local_gradient} is computed by averaging the \emph{clipped} per-record gradient (with clipping bound $R$), which bounds the sensitivity of final aggregation when adding/removing one record from the local dataset.  Second, comparing with \eqref{equ:LFH_server_aggregate}, the server adds Gaussian noise when computing the aggregated global momentum $\tilde{\bm{m}}_t$ in \eqref{equ:server_aggregate} to guarantee DP.  Third, instead of aggregating all clients' momentum, our method also considers aggregating a subset of them, reflected by the index set $\mathcal{I}_t$ in \eqref{equ:server_aggregate}. It provides additional privacy amplification from \emph{client-level} sampling with sampling rate $q$. Note that the original privacy amplification is provided by \emph{record-level} sampling. 

\subsection{Privacy Analysis}
\label{sec:privacy_analysis}
Before presenting the final privacy analysis of DP-BREM, we first show how we compute the sensitivity for the summation of clipped client momentum in \eqref{equ:server_aggregate} for privacy analysis of one iteration, shown in Lemma \ref{lem:aggregation_sensitivity}. We note that clients may have different sizes of local datasets $\mathcal{D}_i$ and can use different record-level sampling rates $p_i$, thus the record-level sensitivity (denoted by $S_i$) for different clients can be different.

\begin{lemma}[DP Sensitivity]
    \label{lem:aggregation_sensitivity}
    We use $\|\cdot\|$ to denote L2-norm $\|\cdot\|_2$. In the $t$-th round, denote the query function $Q_t(\mathcal{D})\coloneqq\sum\nolimits_{j\in\mathcal{I}_t}\mathsf{Clip}_C(\bm{m}_{t,j} - \tilde{\bm{m}}_{t-1})$, where $\tilde{\bm{m}}_{t-1}$ is public and $\mathcal{D}=\{\mathcal{D}_j\}_{j\in\mathcal{I}_t}$. Consider the neighboring dataset  $\mathcal{D}^\prime=\{\mathcal{D}_j\}_{j\neq i, j\in\mathcal{I}_t}\cup \mathcal{D}_i^\prime$ that differs in one record from client $\mathsf{C}_i$'s local data ($i\in\mathcal{I}_t$), i.e., $|\mathcal{D}_i-\mathcal{D}_i^\prime|=1$, then the sensitivity with respect to client $\mathsf{C}_i$ is
    \begin{align}
        \label{equ:aggregation_sensitivity}
        S_i  \coloneqq \max_{\mathcal{D}, D^\prime}\|Q_t(\mathcal{D})-Q_t(\mathcal{D^\prime})\|
        =\min\left\{2C,\frac{R}{p_i|\mathcal{D}_{i}|}\right\}
    \end{align}
\end{lemma}
\begin{proof}
(Sketch) According to \eqref{equ:local_gradient}, the sensitivity of $\bar{\bm{g}}_{t,i}$ is $\frac{R}{p_i|\mathcal{D}_{i}|}$ because each clipped term $\mathsf{Clip}_R(\cdot)$ has bounded L2-norm, i.e., $\|\mathsf{Clip}_R(\cdot)\|\leqslant R$. Then, due to the recursive representation of local momentum in \eqref{equ:local_momentum}, the sensitivity of $\bm{m}_{t,i}$ is $\frac{R}{p_i|\mathcal{D}_{i}|}$. Finally, the client-level clipping $\mathsf{Clip}_C(\cdot)$ introduces another upper bound for the sensitivity. Refer to Appendix \ref{apx:proof_of_lem_aggregation_sensitivity} for the detailed proof.
\end{proof}

\textbf{Remark: comparison with the privacy accountant of DP-SGD momentum.} 
As discussed in Section \ref{sec:challenges_and_baseline}, the privacy accountant of DP-SGD with momentum (i.e., account for privacy cost of gradient, then do post-processing for momentum) requires clients to add noise in the \emph{local gradients}, which leads to poor utility especially when Byzantine attacks exist. In Lemma \ref{lem:aggregation_sensitivity}, we account for the privacy cost of \emph{aggregated momentum}, where the sensitivity is carefully computed from the bounded record-level gradient. Therefore, our scheme solves the three limitations shown in Section \ref{sec:challenges_and_baseline}, which is explained as follows. First, only the server adds noise (which is the same as the central setting), thus the privacy-utility tradeoff is not impacted. Second, the noise is added after the centered clipping $\mathsf{Clip}_C(\bar{\bm{m}}_{t,i} - \tilde{\bm{m}}_{t-1})$, thus it only introduces unbiased error. We also show that (in Section \ref{sec:convergence_analysis}) the impact from the added noise is separate from the impact from Byzantine attacks, as versus the impact from the local noise is enlarged with Byzantine attacks in DP-LFH (see Section \ref{sec:challenges_and_baseline}). Third, since privacy is accounted on momentum, and only the aggregated momentum leaks privacy, our solution enjoys privacy amplification from client-level sampling. 

The final privacy analysis of DP-BREM is shown in Theorem \ref{thm:privacy_analysis}. It presents how to compute the privacy budget $\epsilon$ and privacy parameter $\delta$ when the parameters (such as $T$, $\sigma$, $q$, etc.) of the algorithm are given.  We use an advanced privacy accountant tool called Gaussian DP (GDP) \cite{dong2019gaussian}, then convert it to  $(\epsilon,\delta)$-DP. Note that in our privacy analysis, clients can use different record-level sampling rates $p_i$, thus different sensitivity $S_i$ shown in \eqref{equ:aggregation_sensitivity}. Therefore, the final privacy budget (denoted by $\epsilon_i$) of DP-BREM may be different for different clients, which provides personalized privacy if these parameters are different for each client. 

\begin{theorem}[Privacy Analysis]
    \label{thm:privacy_analysis}
    DP-BREM (in Algorithm \ref{alg:DP_BREM}) satisfies record-level $(\epsilon_i, \delta)$-DP for an honest client $\mathsf{C}_i$ with $\epsilon_i$ and $\delta$ satisfying
    \begin{align}
    \label{equ:delta_epsilon}
	\delta=\Phi\left(-\frac{\epsilon_i}{\mu_i}+\frac{\mu_i}{2}\right)-e^{\epsilon_i}\cdot\Phi\left(-\frac{\epsilon_i}{\mu_i}-\frac{\mu_i}{2}\right),
    \end{align}
    where $\Phi(\cdot)$ denotes the cumulative distribution function (CDF) of standard normal distribution, and $\mu_i$ is defined by 
\begin{align}
    \label{equ:mu_i}
    \mu_i=
    qp_i\sqrt{T(e^{1/(2\sigma_i^2)}-1)},\quad \text{with } \sigma_i = \sigma\cdot\max\left\{\frac{R}{2C},p_i|\mathcal{D}_{i}|\right\}
\end{align}
\end{theorem}
\begin{proof}
This result is obtained by the composition of multiple iterations and the privacy amplification from sampling. See Appendix \ref{apx:proof_of_thm_privacy_analysis} for the detailed proof.
\end{proof}

\subsection{Convergence Analysis}
\label{sec:convergence_analysis}

Before presenting the final convergence analysis of our solution, we first show the aggregation error for one iteration in Theorem \ref{thm:aggregation_error}.

\begin{theorem}[Aggregation Error]
\label{thm:aggregation_error}

Denote $\bm{m}_t^*\coloneqq\frac{1}{|\mathcal{H}|}\sum\nolimits_{i\in\mathcal{H}} \bm{m}_{t,i}$ as the ground truth aggregated raw momentum, where $\bm{m}_{t,i}$ is the client momentum computed from gradient \emph{without} record-level clipping. Assume the local momentum of all honest clients $\{\bm{m}_{t,i}\}_{i\in\mathcal{H}}$ are i.i.d. with expectation $\bm{\mu}\coloneqq\mathbb{E}[\bm{m}_{t,i}]$, and the variance is bounded (in terms of L2-norm)
\begin{align}
    \label{equ:definition_rho}
    \mathbb{E}\|\bm{m}_{t,i}-\bm{\mu}\|^2\leqslant\rho^2
\end{align}
After some parameter tuning (the detailed tuning is shown under \eqref{equ:E_m_t_mu} in Appendix \ref{apx:proof_of_thm_aggregation_error}) of the clipping bounds:
\begin{align}
    \label{equ:tuning_of_R_and_C}
    R\propto O\left(\rho\sqrt{n/(|\mathcal{B}|+\sqrt{d}\sigma/q)}\right), \quad C\propto O(R)
\end{align}
we have the following aggregation error due to clipping, DP noise, and Byzantine clients:
\begin{align}
\label{equ:E_m_t_m_t_informal}
\mathbb{E}\|\tilde{\bm{m}}_{t}-\bm{m}_t^*\|^2 \leqslant O\left(\frac{\rho^2(|\mathcal{B}|+\sqrt{d}\sigma/q)}{n}\right)
\end{align}    
where $|\mathcal{B}|$ is the number of Byzantine clients, $d$ is the dimension of model parameter $\bm{\theta}_t$, $\sigma$ is the noise multiplier (for DP) shown in \eqref{equ:server_aggregate}, $q$ is the client-level sampling rate shown in Line-2 of Algorithm \ref{alg:DP_BREM}, and $\rho$ is defined in \eqref{equ:definition_rho}. The formal version of \eqref{equ:E_m_t_m_t_informal} is shown in \eqref{equ:E_m_t_m_t} of  Appendix \ref{apx:proof_of_thm_aggregation_error}.
\end{theorem}
\begin{proof}
(Sketch) Directly bounding $\mathbb{E}\|\tilde{\bm{m}}_{t}-\bm{m}_t^*\|^2$ is not easy, thus we utilize the upper bounds of $\mathbb{E}\|\tilde{\bm{m}}_{t}-\bm{\mu}\|^2$ and $\mathbb{E}\|\bm{\mu}-\bm{m}_t^*\|^2$ to get the final result, where $\bm{\mu}\coloneqq\mathbb{E}[\bm{m}_{t, i}]$ is the expected local momentum (we assume clients' local momentum are i.i.d.). When upper bounding $\mathbb{E}\|\tilde{\bm{m}}_{t}-\bm{\mu}\|^2$, we decompose errors into three types: honest clients' error (from clipping randomness and bias), Byzantine clients' error (from perturbation), and DP noise error. Optimizing parameters $C$ and $R$ minimizes the total error. See Appendix \ref{apx:proof_of_thm_aggregation_error} for the detailed proof.
\end{proof}

\textbf{Interpretation of Theorem \ref{thm:aggregation_error}.} The value of $\mathbb{E}\|\tilde{\bm{m}}_{t}-\bm{m}_t^*\|^2$ quantifies the aggregation error, i.e., how the aggregated privatized momentum $\tilde{\bm{m}}_{t}$ (with clipping, DP noise, and Byzantine clients' impact) differs from the "pure" momentum aggregation $\bm{m}_t^*$, where only honest clients participate and without clipping and DP noise. According to \eqref{equ:E_m_t_m_t_informal}, the aggregation error is proportional to $\rho^2$ and $\frac{|\mathcal{B}|}{n}+\frac{\sqrt{d}\sigma}{nq}$, where $\rho^2$ quantifies the variance of honest clients' local momentum, $\frac{|\mathcal{B}|}{n}$ is the fraction of Byzantine clients, and  $\frac{\sigma}{nq}=O(1/\epsilon)$ for $\epsilon$-DP. In other words, the aggregation error will be enlarged when: honest clients' variance is large, or the Byzantine attacker corrupts more clients, or the training model is complex (i.e., the model dimension $d$ is large), or we need stronger privacy (i.e., a smaller $\epsilon$), or the number of clients n is small. Furthermore, due to the format of $\frac{|\mathcal{B}|}{n}+\frac{\sqrt{d}\sigma}{nq}$, the impact from DP noise is independent of the increase of Byzantine clients $|\mathcal{B}|$ (versus Limitation 2 of DP-LFH in Section \ref{sec:challenges_and_baseline}). On the other hand, according to the parameter tuning in \eqref{equ:tuning_of_R_and_C},  we could theoretically set a smaller record-level clipping bound $R$ when $\sigma$, $d$, and $|\mathcal{B}|$ are large, or $\rho$ and $n$ are small. The tuning of client-level clipping bound $C$ should be adjusted according to the value of $R$. Recall that $R$ is for DP, while $C$ is for robustness.

By following the convergence analysis in \cite{karimireddy2021learning} and using the result in \eqref{equ:E_m_t_m_t_informal}, we have the convergence rate shown below.

\begin{theorem}[Convergence Rate of DP-BREM]
\label{thm:convergence_rate}
The convergence rate of DP-BREM in Algorithm \ref{alg:DP_BREM} is asymptotically (ignoring constants and higher order terms) of the order
\begin{align}
    \label{equ:convergence_rate}
    \frac{1}{T}\sum\nolimits_{t=1}^T\mathbb{E}\|\nabla\ell(\bm{\theta}_{t-1})\|^2 \lesssim \sqrt{\frac{\rho^2}{T}\frac{|\mathcal{B}|+(1+\sqrt{d}\sigma)/q}{n}}
\end{align}
where $\ell(\cdot)$ is the loss function, $T$ is the total number of training iterations, and other parameters are the same as in \eqref{equ:E_m_t_m_t_informal}.
\end{theorem}
\begin{proof}
See Appendix \ref{apx:proof_of_thm_convergence_rate}.
\end{proof}

\textbf{Remark: comparison with LFH and DP-LFH.} 
The convergence rate of the non-private LFH, DP-LFH, and the proposed solution DP-BREM, showing in \eqref{equ:convergence_rate_of_LFH}, \eqref{equ:convergence_rate_of_baseline}, and \eqref{equ:convergence_rate} respectively, are summarized in Table \ref{tab:convergence_rate}. Though both DP-LFH and DP-BREM pay an additional term of $\sqrt{d}\sigma/q$ to get the DP property, they have different impacts on the convergence. As discussed in Limitation 2 of Section \ref{sec:challenges_and_baseline}, the additional term  $\sqrt{d}\sigma/q$ of DP-LFH (due to DP noise added to clients' gradient) is on the term $\rho$, thus it will enlarge the impact of Byzantine clients (i.e., the term $|\mathcal{B}|$). However, the additional term  $\sqrt{d}\sigma/q$ of our solution DP-BREM (due to DP noise added to the aggregated momentum) is on the term $1+|\mathcal{B}|$, which has a squared-root order. Therefore, DP noise only has a limited impact on the convergence of DP-BREM when there are Byzantine clients. We will validate the above theoretical analysis via experimental results in Section \ref{sec:experiments}.

\begin{table}[!t]
\centering
\small
\caption{Comparison of Convergence Rate}
\begin{tabular}{c|cc}
\hline
    & Where to add noise  & Convergence Rate \\
\hline 
    LFH \cite{karimireddy2021learning} & None & $O(\rho\sqrt{1+|\mathcal{B}|})$\\
    DP-LFH & Clients' gradients & $O\left((\rho+\sqrt{d}\sigma)\sqrt{1+|\mathcal{B}|}\right)$ \\
    DP-BREM & Aggregated momentum & $O\left(\rho\sqrt{1+|\mathcal{B}|+\sqrt{d}\sigma}\right)$\\
\hline
\end{tabular}
\label{tab:convergence_rate}
\end{table}

\section{DP-BREM\textsuperscript{+} with Secure Aggregation}
\label{sec:secure_aggregation}

The private and robust FL solution DP-BREM (in Section \ref{sec:our_solution}) assumes a \emph{trusted} server which can access clients' momentum. In this section, we propose DP-BREM\textsuperscript{+}, which assumes a \emph{malicious} server and utilizes secure aggregation techniques, achieving the same DP and robustness guarantees as DP-BREM. As discussed in Section \ref{sec:problem_statement}, we consider the server as malicious only for data privacy, while clients are malicious for both data privacy and Byzantine attacks.

\subsection{Challenges}
\label{sec:secure_aggregation_problem_overview}

Considering the server is malicious for data privacy, the noisy aggregate of momentum with centered clipping shown in \eqref{equ:server_aggregate} must be implemented securely with the goals of 1) privacy, i.e., each party, including clients and the server, learns nothing but the differentially-private output; and 2) integrity, i.e., the output is correctly computed. Since the noisy aggregated momentum of the previous iteration $\tilde{\bm{m}}_{t-1}$ already satisfies DP, we can regard it as public information and only need to focus on securely computing the term  $\sum\nolimits_{i\in\mathcal{I}_t}\mathsf{Clip}_C(\bar{\bm{m}}_{t,i} - \tilde{\bm{m}}_{t-1}) + \mathcal{N}(0, R^2\sigma^2 \mathbf{I}_d)$ in \eqref{equ:server_aggregate}.

\textbf{Secure Aggregation with Verified Inputs (SAVI).}
The key crypto technique we leverage to achieve the above objectives is  SAVI \cite{roy2022eiffel}, which is a type of protocols that securely aggregate only well-formed inputs. The security goals include both \emph{privacy} and \emph{integrity}. Specifically, privacy means that no party should be able to learn anything about the raw input of an honest client, other than what can be learned from the final aggregation result. Integrity means that the output of the protocol returns the correct aggregate of well-formed input, where 1) an input $u$ passes the input integrity check with a public validation predicate $\mathsf{Valid}(\cdot)$ if and only if $\mathsf{Valid}(u)=1$, and 2) the aggregation is correctly computed. An instantiation of the SAVI protocol is EIFFeL \cite{roy2022eiffel} (described in Appendix \ref{apx:secure_aggregation_preliminaries}).

\textbf{Challenge: Secure Generation of Gaussian Noise.}
A SAVI protocol can potentially solve the problem of securely aggregating the clipped vectors (by enforcing a norm-bound on the client momentum difference). However, the Gaussian noise $\mathcal{N}(0, R^2\sigma^2 \mathbf{I}_d)$ needs to be securely generated and aggregated as well. In DP-BREM with a trusted server, the Gaussian noise $\mathcal{N}(0, R^2\sigma^2 \mathbf{I}_d)$ is generated by the server to guarantee DP. However, when the server is assumed as malicious, the added Gaussian noise for DP cannot be directly generated by the server.  


A straightforward approach is to use a semi-honest server, as proposed in \cite{roy2020crypt}, to generate DP noise and manage the privacy engine. However, relying on another non-colluding server may be impractical, so we assume only a single server. An alternative is Distributed DP \cite{shi2011privacy}, where clients locally generate Gaussian noise. The aggregated noise follows a Gaussian distribution with an enlarged standard deviation, ensuring DP through cryptographic techniques. This method, however, has two limitations: it requires more noise to achieve the same privacy level due to potential collusion among malicious clients, and the robustness is compromised as malicious clients can generate arbitrary local noise.



A possible solution to address the first limitation is to \emph{jointly} generate Gaussian noise as in \cite{pentyala2022training}, where no party learns or controls the true value of the noise (or a portion of the noise). However, the protocol in \cite{pentyala2022training} is designed only for additive secret sharing schemes, which only works for honest-but-curious parties and does not tolerate malicious parties. Moreover, in \cite{pentyala2022training}, the Gaussian noise is jointly generated by honest-but-curious and non-colluding parties, which does not address the second limitation as the clients can be malicious in our threat model discussed in Section \ref{sec:problem_statement}.

\textbf{Overview of DP-BREM\textsuperscript{+}.}
To achieve secure aggregation with verified inputs and secure Gaussian noise generation under the threat model of a malicious server and malicious minority of clients, our DP-BREM\textsuperscript{+} 1) leverages an existing SAVI protocol called EIFFeL \cite{roy2022eiffel} to achieve secure input validation; and 2) introduces a new protocol to achieve secure noise generation that is compatible with EIFFeL. The idea of \emph{jointly} generating Gaussian noise in DP-BREM\textsuperscript{+} is inspired by \cite{pentyala2022training}, but our design is based on Shamir's secret sharing \cite{shamir1979share} with robust reconstruction by following the design in EIFFeL, thus guarantees security under malicious minority. We present the preliminaries of Shamir's secret sharing and EIFFeL protocol in Appendix \ref{apx:secure_aggregation_preliminaries}.

\subsection{Design of DP-BREM\textsuperscript{+}}
\label{sec:design_secure_aggregation}

As discussed in Section \ref{sec:secure_aggregation_problem_overview}, the main task of DP-BREM\textsuperscript{+} is to securely compute the term $\sum\nolimits_{i\in\mathcal{I}_t}\mathsf{Clip}_C(\bar{\bm{m}}_{t,i} - \tilde{\bm{m}}_{t-1}) + \mathcal{N}(0,R^2\sigma^2 \mathbf{I}_d)$ shown in \eqref{equ:server_aggregate}. After computing local momentum $\bar{\bm{m}}_{t,i}$ via \eqref{equ:local_momentum}, each client $\mathsf{C}_i$ first implements centered clipping to get $\bm{z}_i\coloneqq \mathsf{Clip}_C(\bar{\bm{m}}_{t,i} - \tilde{\bm{m}}_{t-1})$, which is the private input for validation and aggregation. 

\textbf{Three-Phase Design.} 
In DP-BREM\textsuperscript{+}, clients and the server jointly implement three phases: 1) secure input validation to validate the client momentum is properly centered clipped by $C$, 2) secure noise generation, where clients generate shares of Gaussian noise which can be aggregated in Phase 3 to ensure DP, and 3) aggregation of valid inputs and noise to obtain the noisy global model. We assume the arithmetic circuit is computed over a finite field $\mathbb{F}_{2^K}$. The illustration of DP-BREM\textsuperscript{+} is shown in Figure \ref{fig:framework_secure_aggregation}. Due to limited space, we present the detailed steps \textcircled{\small 1}-\textcircled{\small 7} in Appendix \ref{apx:detailed_steps_of_secure_aggregation}.

\textbf{Phase 1: Secure Input Validation.} 
The validation function for an input $\bm{z}_i$ considered in DP-BREM\textsuperscript{+} is defined as $\mathsf{Valid}(\bm{z}_i)\coloneqq \mathbbm{1}(\|\bm{z}_i\|\leqslant C)$, where $\mathsf{Valid}(\bm{z}_i)=1$ if and only if the condition $\|\bm{z}_i\|\leqslant C$ holds. Since honest clients compute $\bm{z}_i= \mathsf{Clip}_C(\bar{\bm{m}}_{t,i} - \tilde{\bm{m}}_{t-1})$, verifying whether $\bm{z}_i$ is well-formed, with bounded L2-norm via $\mathsf{Valid}(\cdot)$, for all clients ensures centered clipping of client momentum $\bar{\bm{m}}_{t,i}$ (to achieve robustness as DP-BREM). We follow the design in EIFFeL \cite{roy2022eiffel} for secure input validation, which returns the validation result $\mathsf{Valid}(\bm{z}_i)$ (either 1 or 0) for client $\mathsf{C}_i$'s private input $\bm{z}_i$, corresponding to steps \textcircled{\small 1}, \textcircled{\small 2}, and \textcircled{\small 3} shown in Figure \ref{fig:framework_secure_aggregation}. Then, clients and the server can jointly verify all inputs $\{\bm{z}_i\}_{i\in\mathcal{I}_t}$, and obtain the set of valid inputs $\mathcal{I}_{\mathsf{Valid}}$, where $\mathsf{Valid}(\bm{z}_i)=1$ for all $i\in\mathcal{I}_{\mathsf{Valid}}$. In the later step, only inputs in $\mathcal{I}_{\mathsf{Valid}}$ are aggregated.

\textbf{Phase 2: Secure Noise Generation.}
We develop a new protocol for secure distributed Gaussian noise generation, which returns the shares (held by each client) of a random vector $\bm{\xi}$ of length $d$ from the Gaussian distribution $\mathcal{N}(0,R^2\sigma^2 \mathbf{I}_d)$, corresponding to steps \textcircled{\small 4} and \textcircled{\small 5} shown in Figure \ref{fig:framework_secure_aggregation}. The shares of noise can be reconstructed into a single Gaussian noise (for ensuring DP) with the guarantee that no parties know or control the generated noise, which protects the information of private inputs after the noisy aggregate is released.

\textbf{Phase 3: Aggregation of Valid Inputs and Noise.}
Finally, the server and clients can aggregate the valid inputs (obtained in Phase 1) and the generated Gaussian noise (obtained in Phase 2) by implementing steps \textcircled{\small 6} and \textcircled{\small 7} shown in Figure \ref{fig:framework_secure_aggregation}, ensuring nothing except the noisy aggregate can be learned.

\begin{figure}[!t]
    \centering
    \includegraphics[width=3.3in]{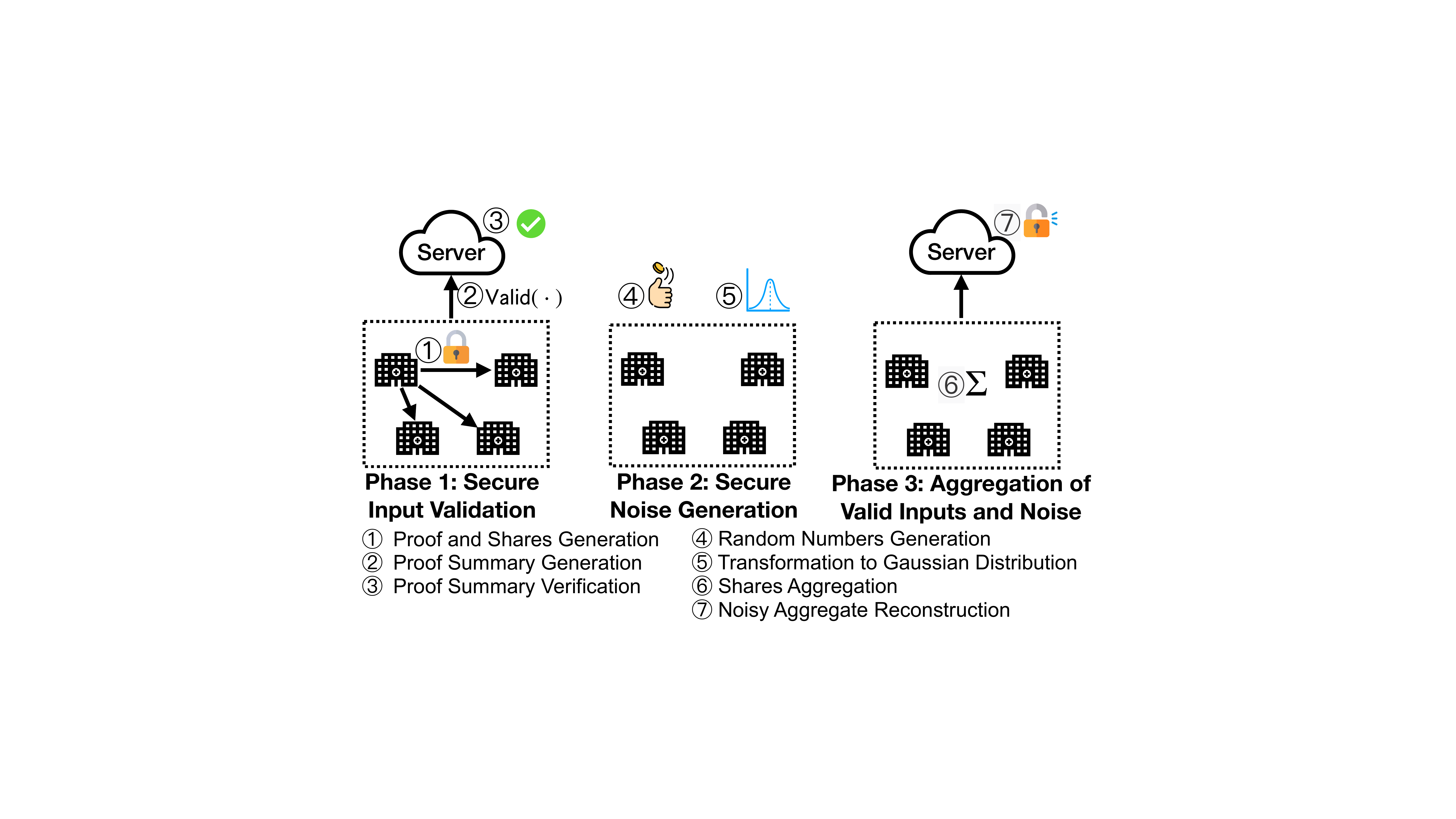}
    \vspace{-3mm}
    \caption{Illustration of DP-BREM\textsuperscript{+} (see Appendix \ref{apx:detailed_steps_of_secure_aggregation} for detailed steps \textcircled{\small 1}-\textcircled{\small 7})}
    \vspace{-3mm}
    \label{fig:framework_secure_aggregation}
\end{figure}

\textbf{Remark on Efficiency.}
DP-BREM\textsuperscript{+}'s usage of EIFFeL's secure input validation is due to efficiency considerations. Instead of having clients perform clipping and using secure input validation, one alternative is to use standard secure multi-party computation (MPC) for the clipping and aggregation. However, doing this under MPC would result in a very large computation/communication overhead due to the multiplication, min-operation, division, and L2-norm computation in the clipping operation $\mathsf{Clip}_C(\cdot)$ defined in \eqref{equ:clipping_function}. In contrast, the secure input validation protocol only requires the verifiers to check all the multiplication gates very efficiently with just one identity test.  The compatibility with secure input validation is one of the advantages of  DP-BREM. 

\textbf{Complexity.}
According to EIFFeL \cite{roy2022eiffel}, the computation/communication complexity of secure aggregation with input validation is $O(mnd)$ for clients and $O(n^2+md\min\{n,m^2\})$ for the server in terms of the number of clients $n$, number of malicious clients $m$, and data dimension $d$. For the proposed secure noise generation (only clients are involved), the computation/communication complexity for total $n$ clients is  $O(mnd)$.

\begin{figure*}[!t]
    \centering
    \includegraphics[width=\linewidth]{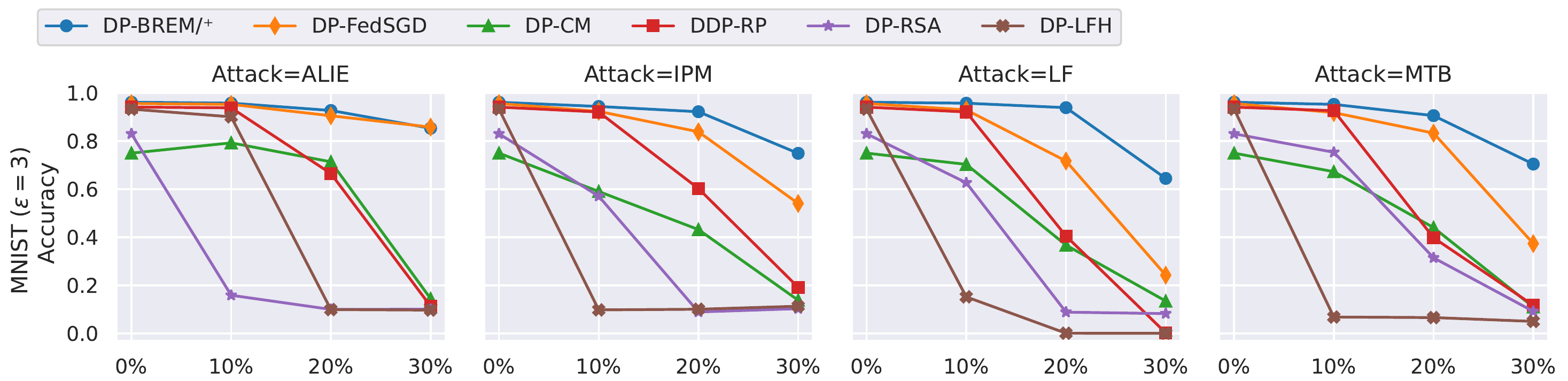}
    \includegraphics[width=\linewidth]{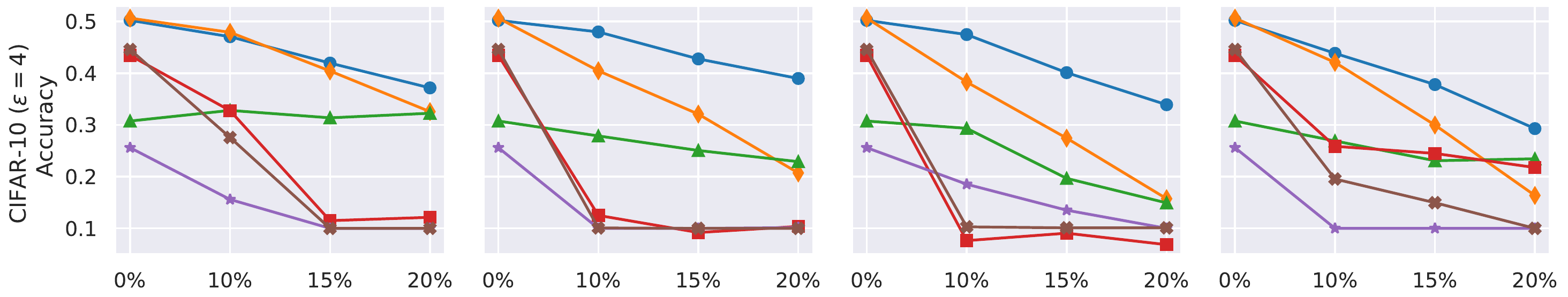}
    \includegraphics[width=\linewidth]{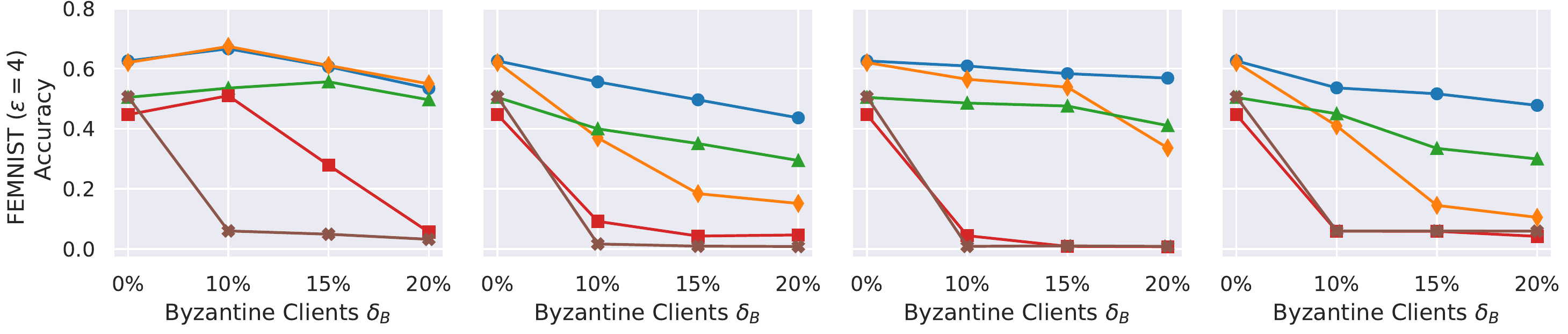}
    \vspace{-8mm}
    \caption{With fixed privacy budget $\epsilon$, varying the percentage of Byzantine clients $\delta_B$ for three datasets.}
    \vspace{-4mm}
    \label{fig:vary_K}
\end{figure*}

\subsection{Security Analysis}
\label{sec:security_analysis}
In comparison, EIFFeL \cite{roy2022eiffel} is a secure aggregation protocol with verified inputs (without guaranteeing DP), while our solution DP-BREM\textsuperscript{+} is a secure noisy aggregation protocol with verified inputs and jointly generated Gaussian noise, which provides DP on the aggregated results. Therefore, the only difference is the Gaussian noise that will be aggregated to the final result. We show the formal security guarantee of DP-BREM\textsuperscript{+} in the following theorem.

\begin{theorem}[Security Guarantees of DP-BREM\textsuperscript{+}]
\label{thm:security_analysis}
For the validation function $\mathsf{Valid}(\cdot)$ considered in Section \ref{sec:design_secure_aggregation}, given a security parameter $\kappa$, the secure noisy aggregation protocol in DP-BREM\textsuperscript{+} satisfies:

\textbf{1) Integrity.}
For a negligible function $\text{negl}(\cdot)$, the output of the protocol returns the noisy aggregate of a subset of clients $\mathcal{I}_{\mathsf{Valid}}$ and Gaussian noise $\bm{\xi}$, such that all clients in $\mathcal{I}_{\mathsf{Valid}}$ have well-formed inputs:
\begin{align*}
    \Pr[\text{output}=\sum\nolimits_{i\in\mathcal{I}_{\mathsf{Valid}}} \bm{z}_i + \bm{\xi}] \geqslant 1-\text{negl}(\kappa)
\end{align*}
where random vector $\bm{\xi}\sim\mathcal{N}(0,R^2\sigma^2\mathbf{I}_d)$, and $\mathsf{Valid}(\bm{z}_i)=1$ for all $i\in\mathcal{I}_{\mathsf{Valid}}$. Note that the set $\mathcal{I}_{\mathsf{Valid}}$ contains all honest clients (denoted by $\mathcal{I}_H$) and the malicious clients who submitted well-formed input (denoted by $\mathcal{I}_M^*$), i.e., $\mathcal{I}_{\mathsf{Valid}}=\mathcal{I}_H\cup\mathcal{I}_M^*$.

\textbf{2) Privacy.} For a set of malicious clients $\mathcal{I}_M$ and a malicious server $\mathsf{S}$, there exists a probabilistic polynomial-time (P.P.T.) simulator $\mathsf{Sim}(\cdot)$ such that:
\begin{align*}
    \mathsf{Real}\left(\{z_i\}_{i\in\mathcal{I}_H}, \Omega_{\mathcal{I}_M\cup\mathsf{S}}\right) \equiv_{\mathsf{C}} \mathsf{Sim}\left(\sum\nolimits_{i\in\mathcal{I}_{H}} \bm{z}_i + \bm{\xi},\mathcal{I}_H, \Omega_{\mathcal{I}_M\cup\mathsf{S}}\right)
\end{align*}
where $\{z_i\}_{i\in\mathcal{I}_H}$ denotes the input of all the honest clients, $\mathsf{Real}$ denotes a random variable representing the joint view of all the parties in the protocol's execution, $\Omega_{\mathcal{I}_M\cup\mathsf{S}}$ indicates a polynomial-time algorithm implementing the "next-message" function of the parties in $\mathcal{I}_M\cup\mathsf{S}$ (see \cite[Appendix~11.5]{roy2022eiffel}), and $\equiv_{\mathsf{C}}$ denotes computational indistinguishability. In summary, the server and clients learn nothing besides the final aggregated result.
\end{theorem}
\begin{proof}
See Appendix \ref{apx:proof_of_thm_security_analysis}.
\end{proof}

\section{Experimental Evaluation}
\label{sec:experiments}

\begin{figure*}[!t]
    \centering
    \includegraphics[width=\linewidth]{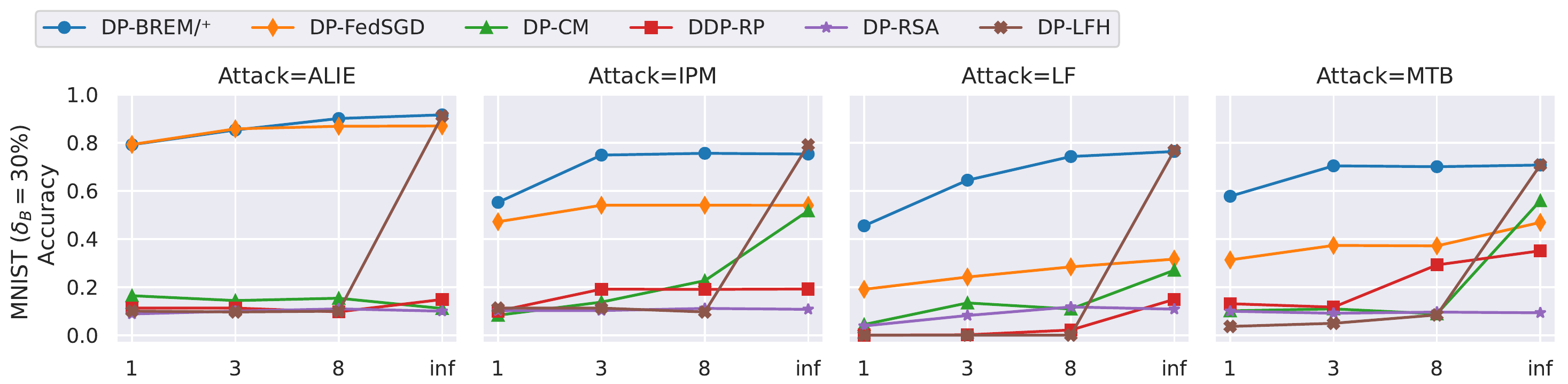}
    \includegraphics[width=\linewidth]{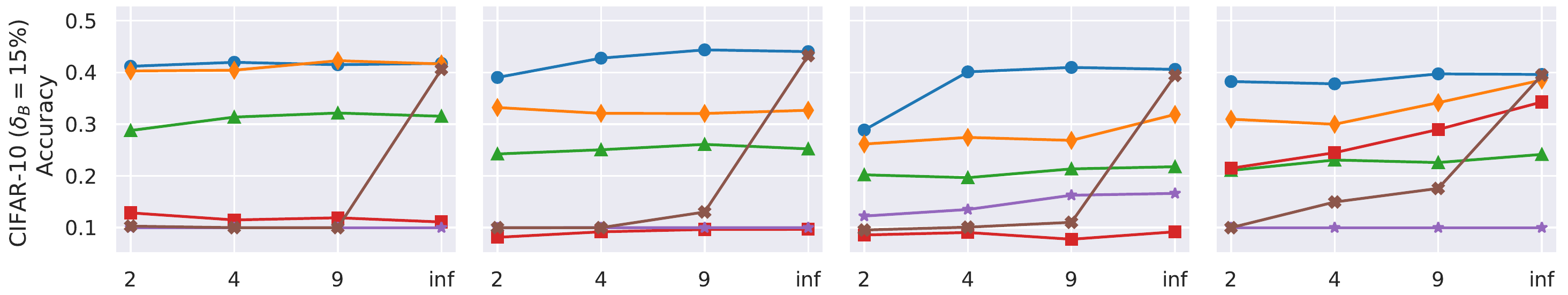}
    \includegraphics[width=\linewidth]{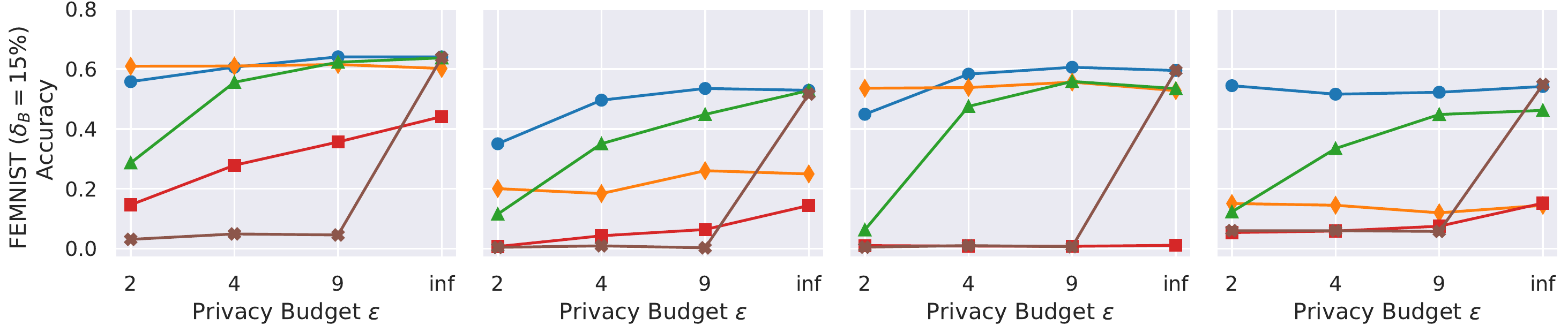}
    \vspace{-8mm}
    \caption{With fixed percentage of Byzantine clients $\delta_B$, varying privacy budget $\epsilon$ for three datasets.}
    \vspace{-4mm}
    \label{fig:vary_epsilon}
\end{figure*}

In this section, we demonstrate the effectiveness of the proposed DP-BREM/DP-BREM\textsuperscript{+} on achieving both good privacy-utility tradeoff and Byzantine robustness via experimental results on MNIST \cite{lecun1998mnist}, CIFAR-10 \cite{krizhevsky2009learning}, and FEMNIST \cite{caldas2018leaf} datasets with non-IID setting (refer to Appendix \ref{apx:experimental_setup} for more details on the datasets and model architectures). Note that MNIST and CIFAR-10 have 10 classes, while FEMNIST includes 62 classes. All experiments are developed via PyTorch\footnote{Our source code is available at https://github.com/xiaolangu/DP-BREM}.

\textbf{Byzantine Attacks.}
We consider four existing Byzantine attacks in our experiments, including ALIE ("a little is enough") \cite{baruch2019little}, IPM (inner-product manipulation) \cite{xie2020fall}, LF (label-flipping), and the state-of-the-art  MTB ("manipulating-the-Byzantine") \cite{shejwalkar2021manipulating}. Refer to Appendix \ref{apx:experimental_setup} for more details.

\textbf{Compared Methods.}
We compare the performance of six approaches against Byzantine attacks, including DP-BREM/\textsuperscript{+} (our approach)\footnote{Since DP-BREM\textsuperscript{+} achieves the same DP and robustness guarantees as DP-BREM, we did not perform the empirical experiments with secure aggregation because the accuracy results will be exactly the same as DP-BREM. We use DP-BREM/\textsuperscript{+} to denote both DP-BREM and DP-BREM\textsuperscript{+}, and the implementation follows Algorithm \ref{alg:DP_BREM}.}, a variant of DP-FedSGD \cite{mcmahan2018learning} with both record and client norm clipping, DDP-RP \cite{wang2022privacy}, DP-RSA \cite{zhu2022bridging}, a variant of CM \cite{yin2018byzantine} with DP noise, and DP-LFH. The comparison (on trust assumption and mechanism overview) of these approaches is provided in Table \ref{tab:DP_compare}, and Appendix \ref{apx:experimental_setup} shows more details of each approach. In summary,  DP-BREM/\textsuperscript{+}, DP-FedSGD, and DP-CM add central noise to the aggregation, but DP-BREM\textsuperscript{+} does not require a trusted server due to the secure aggregation technique. DDP-RP adds \emph{partial} local noise to the client's update with secure aggregation. DP-RSA and DP-LFH add local noise to the client's update. We fix $\delta=10^{-6}$ for $(\epsilon, \delta)$-DP in all experiments. For the setting of other parameters, refer to Appendix \ref{apx:parameters_in_experiments}.

\textbf{Evaluation Metric.}
We evaluate the testing accuracy of the global model within $T$ iterations. Considering the accuracy curve might be unstable under Byzantine attacks, we average the accuracy between $0.9T$ and $T$ as the final accuracy for comparison. Note that both DP noise and Byzantine attacks reduce the accuracy. A protocol achieves good Byzantine robustness if its accuracy does not decrease too much with an increased number of Byzantine clients.

\subsection{Robustness Evaluation with DP}

We consider a fixed privacy budget $\epsilon$ and implement each attack with different percentages of Byzantine clients $\delta_B=\frac{|\mathcal{B}|}{n}$ for the four attacks, and compare the accuracy among all approaches. We note that a complex dataset requires a more sophisticated model architecture and makes it more challenging to maintain good utility in the presence of DP and Byzantine attacks. Therefore, in our experiments with CIFAR-10 (which has three color channels) and FEMNIST (which includes 62 classes), we use slightly larger $\epsilon$ values and a smaller number of Byzantine clients. These choices are still within a reasonable range. Previous papers, such as \cite{abadi2016deep} and \cite{zheng2021federated}, also used larger privacy budgets for the CIFAR-10 dataset compared to the MNIST dataset. The results for MNIST (with $\epsilon=3$), CIFAR-10 (with $\epsilon=4$), and FEMNIST (with $\epsilon=4$) datasets are shown in Figure \ref{fig:vary_K}. Compared to the results on the MNIST dataset (with 10 classes), the accuracy on the FEMNIST dataset is lower due to the larger number of classes. Though the detailed results vary under different attacks and across three datasets, we have some general observations:

1)  When there is no attack, i.e., $\delta_B=0$, DP-BREM/\textsuperscript{+} achieves almost the same accuracy as DP-FedSGD, indicating the Byzantine-robust design (client momentum with centered clipping) has almost no impact on the utility in this case.

2) After increasing $\delta_B$, our DP-BREM/\textsuperscript{+} has the smallest accuracy decrease, indicating its success in providing Byzantine robustness. However, the accuracy of DP-LFH reduces sharply, demonstrating that the large aggregated local DP noise makes the robust aggregator more vulnerable to Byzantine attacks, which is consistent with our discussions of Limitation 2 in Section \ref{sec:challenges_and_baseline}. 

3) Though DP-FedSGD has client-level gradient clipping,  which can restrict malicious clients' impact, it is still vulnerable to some types of Byzantine attacks (such as IPM and MTB) under larger $\delta_B$ values.

4) CM with DP noise (or DP-CM) has a relatively small accuracy decrease for a relatively small $\delta_B$. It is the benefit of the median-based robust aggregator. But the sensitivity is larger than the average-based aggregators, as discussed in Example \ref{exp:sensitivity_average_median}, the aggregated DP noise is too large to obtain a high accuracy, even when $\delta_B=0$.

5) DDP-RP is more vulnerable to LF attack because it only checks the element-wise range. Also, the model replacement strategy in LF attack is more likely to change the positions that have small values in benign gradient vectors.

6) DP-RSA has relatively poor accuracy compared to other approaches, even when $\delta_B=0$. This is caused by the sign-SGD aggregator, which only aggregates element-wise \emph{signs} rather than the full precision gradient, leading to large information loss. Moreover, the local DP noise makes Byzantine attacks easier to succeed. We note that DP-RSA does not converge for the FEMNIST dataset (possibly caused by the sign aggregation), even without DP noise and Byzantine clients, and thus we do not present the results for this dataset.

7) Under the ALIE attack, it is possible for a small number of Byzantine clients to improve the accuracy of the model compared to the scenario without any Byzantine clients. For instance, an ALIE attack with $10\%$ Byzantine clients can achieve higher accuracy than the case with $0\%$ Byzantine clients across all defense aggregators except DP-LFH. This improvement occurs because the ALIE attack generates malicious gradients that are close to the averaged good gradients but deviate slightly using a scaling factor. This factor is determined based on the total number of clients and the proportion of Byzantine clients, designed to bypass any anomaly detection mechanisms employed by the central server. Consequently, when the number of Byzantine clients is relatively small, the malicious gradients can enhance model accuracy compared to benign gradients, where the record-level clipping is used to achieve DP.

\begin{figure*}[!t]
    \centering
    \includegraphics[width=\linewidth]{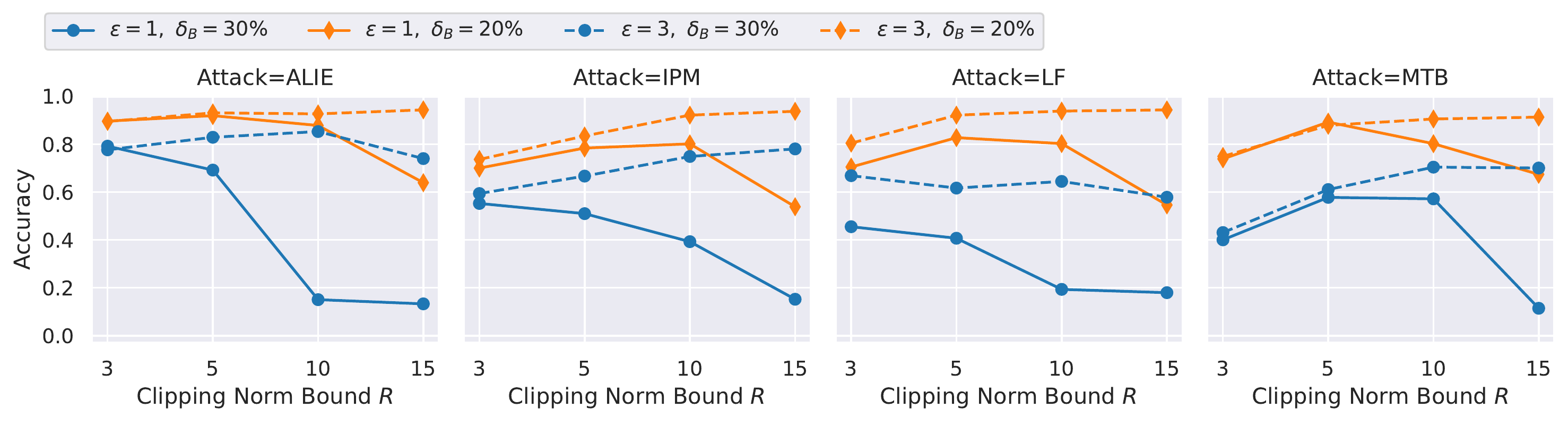}
    \vspace{-8mm}
    \caption{MNIST: Varying record-level clipping bound $R$ for DP-BREM under different settings.}
    \vspace{-4mm}
    \label{fig:mnist_vary_R}
\end{figure*}

\subsection{Privacy-Utility Tradeoff under Attack}
We consider a fixed percentage of Byzantine clients $\delta_B$ for each attack under different values of privacy budget $\epsilon$, and compare the accuracy of all approaches. The results for MNIST (with $\delta_B=30\%$), CIFAR-10 (with $\delta_B=15\%$), and FEMNIST (with $\delta_B=15\%$) datasets are shown in Figure \ref{fig:vary_epsilon}. For all three datasets, we consider four different levels of privacy, where $\epsilon=\inf$ means the standard deviation of DP noise is 0. However, we still implement record-level clipping to illustrate how the noise affects the results while keeping other settings including the clipping step the same.

It's essential to highlight that while the privacy-utility curve is generally monotonic in the absence of Byzantine attacks, this may not hold under Byzantine attacks due to two sources of perturbation. When malicious perturbation dominates, the impact of DP noise on utility is typically minimal. Additionally, different defense aggregators exhibit varying sensitivities to malicious perturbation and DP noise across various datasets, even when the number of malicious clients and $\epsilon$ values are the same. Consequently, observations can vary across defense aggregators, attacks, and datasets (with different parameters). For example, DP noise has a very small (or almost negligible) impact on DP-FedSGD compared to DP-BREM. This could be because DP-FedSGD aggregates more information than the momentum-based solution, leading to a better signal-to-noise ratio (SNR) and thus greater robustness to DP noise. However, the attack has a more significant impact on DP-FedSGD.

Though the detailed results vary under different attacks and across the three datasets, DP-BREM/\textsuperscript{+} generally achieves the highest accuracy among almost all approaches, especially under IPM and MTB attacks. The only exception is when $\epsilon=2$ for the FEMNIST dataset, where the accuracy of DP-BREM/\textsuperscript{+} is lower than that of DP-FedSGD. This is because the client momentum in DP-BREM/\textsuperscript{+} restricts the information that can be learned from each new iteration, making the increased DP noise have a greater impact on the model's accuracy.

Note that when $\sigma=0$ (i.e., $\epsilon=\inf$), both DP-BREM/\textsuperscript{+} and DP-LFH reduce to LFH, thus they have the same results in this case. We can observe that with a moderate privacy budget, such as $\epsilon\geqslant 2$, DP noise only has a negligible impact on the accuracy. But if $\epsilon$ is too small, such as $\epsilon=1$ for the MNIST dataset in Figure \ref{fig:vary_epsilon}, DP-BREM/\textsuperscript{+} suffers a relatively larger impact (but still acceptable) from DP noise. Note that when there exist Byzantine attacks, reducing the DP noise to $\sigma=0$ (i.e., $\epsilon=\inf$)  does not significantly improve the accuracy of DP-BREM/\textsuperscript{+} compared with $\epsilon<\inf$, because Byzantine clients' perturbations largely impact the performance. However, the accuracy of DP-LFH is greatly reduced when $\epsilon<\inf$, since the local DP noise impacts the robustness of the aggregator. This observation is consistent with our theoretical analysis in Limitation 2 of DP-LFH (Section \ref{sec:challenges_and_baseline}).

\begin{table}
\begin{threeparttable}
    \small
    \centering
    \caption{Running time\tnote{\emph{1}} (in milliseconds) per round per client on the MNIST dataset. }
    \begin{tabular}{c|ccc|c}
    \hline
    \makecell[c]{Batch \\ Size}   & \makecell[c]{Baseline \\ {\footnotesize(FedSGD)}} & \makecell[c]{FedSGD+DP \\ {\footnotesize(efficient\tnote{\emph{2}}~~)}} & \makecell[c]{DP-BREM \\ {\footnotesize(DP+robust)}} & \makecell[c]{FedSGD+DP \\ {\footnotesize(inefficient\tnote{\emph{3}}~~)}}\\
    \hline
    30     & 11.80  & 13.31 & 13.72 & 41.06\\
    60     & 18.23 & 19.79 & 20.27 & 76.70\\
    120     & 31.22 & 33.18 & 33.70 & 149.32\\
    \hline
    \end{tabular}
    \begin{tablenotes}
    \footnotesize
        \item[\emph{1}] Our GPU device is NVIDIA Tesla P100-PCIE-16GB. Using other GPU devices may have different results.
        \item[\emph{2}] By default, our implementation uses \emph{efficient} per-record gradient clipping by following Opacus library's implementation with parallel clipping and optimized einsum  (refer to \url{https://opacus.ai/api/_modules/opacus/optimizers/optimizer.html#DPOptimizer})
        \item[\emph{3}] To illustrate the improvement of \emph{efficient} clipping, we also show the results of the \emph{inefficient} implementation, which clips per-record gradient sequentially and without using optimized einsum.
    \end{tablenotes}
    \vspace{-4mm}
    \label{tab:efficiency_result}
    \end{threeparttable}
\end{table}

\subsection{Other Results}

\textbf{Efficiency Evaluation of DP and Byzantine Robustness.}
We note that DP and Byzantine Robustness designs in our solution only introduce a small computation overhead, because 1) the clipping step of DP can be implemented efficiently; 2) our robustness is essentially a clipped summation of client momentum without any complex computations.  Due to limited resources, we implemented the distributed training of FL on a single machine (by running all the clients and the server code sequentially). We evaluate the efficiency of DP-BREM via the running time (per round per client) on the MNIST dataset. The results shown in Table \ref{tab:efficiency_result} indicate that the DP noise and Byzantine robustness only incur $8\%\sim 16\%$ additional running time (depending on batch size).

\begin{table}
    \centering
    \small
    \caption{Model accuracy when varying $C$ of DP-BREM/\textsuperscript{+} with $\epsilon=4$ under IPM and MTB attacks on FEMNIST dataset.} 
    \vspace{2mm}
    \begin{tabular}{c|ccccc}
    \hline
      $\delta_B$   & $C=0.5$ & $C=1$  & $C=2$ & $C=3$ & $C=4$\\
    \hline
     0\%  & 0.622  & \textbf{0.647} & 0.625 & 0.621 &  0.627 \\
     IPM 10\%    & 0.407 & 0.524 & \textbf{0.555} & 0.528 & 0.514\\
     IPM 20\%    & 0.060  & 0.305  & \textbf{0.436} & 0.413 & 0.392 \\
     MTB 10\%    & 0.591 & \textbf{0.605} & 0.535 & 0.525 & 0.545\\
     MTB 20\%    & \textbf{0.554}  & 0.537  & 0.477 & 0.426 & 0.426 \\
    \hline
    \end{tabular}
    \vspace{-4mm}
    \label{tab:vary_C}
\end{table}

\textbf{Impact of $R$ in DP-BREM/\textsuperscript{+}.}
Figure \ref{fig:mnist_vary_R} shows how the accuracy changes w.r.t. the record-level clipping bound $R$ in DP-BREM/\textsuperscript{+}. The results demonstrate that when there are fewer Byzantine clients (i.e., smaller $\delta_B$) or the noise multiplier $\sigma$ is smaller (i.e., larger $\epsilon$), we need to set a larger $R$ to obtain better accuracy. This observation is consistent with the theoretical analysis of parameter tuning discussed in Theorem \ref{thm:aggregation_error} and its interpretation. 

\textbf{Impact of $C$ in DP-BREM/\textsuperscript{+}.}
We use the fixed client-level clipping bound $C$ for each dataset in previous experiments. Table \ref{tab:vary_C} illustrates how varying values of $C$ (while keeping the default and fixed $R$) can influence model accuracy. In the absence of Byzantine attacks, the value of $C$ has a relatively small impact on the model accuracy. However, in the presence of Byzantine attacks, the effect of $C$ varies depending on the nature of the attack. For instance, attacks like the IPM attack, which deviate significantly from benign gradients, benefit from a slightly larger $C$ as it allows more useful information (from benign clients) to be retained. Conversely, for attacks like the MTB attack, which aim to evade detection by aligning more closely with benign gradients, a slightly smaller $C$ can improve accuracy by reducing the impact of the attack on the aggregated gradient.

\textbf{Impact of $q$ in DP-BREM/\textsuperscript{+}.}
In previous experiments, we set client-level sampling rate $q=1$ by default. As discussed in Sec. \ref{sec:our_algorithm}, aggregating a subset $I_t$ of clients in \eqref{equ:server_aggregate} is one of the major differences from LFH. Table \ref{tab:vary_q} demonstrates the utility improvement by optimizing $q$ under different attack percentages $\delta_B$. Intuitively, without attacks, a smaller $q$ enhances privacy amplification, reducing the required $\sigma$ for a given $\epsilon$ in DP; however, too small a $q$ increases aggregation variance. Under Byzantine attacks, a smaller $q$ mitigates attack impact as only a subset of Byzantine clients are aggregated. Thus, with higher $\delta_B$ the optimal $q$ (highlighted in Table \ref{tab:vary_q}) decreases.

\begin{table}
    \centering
    \small
    \caption{Model accuracy when varying $q$ of DP-BREM/\textsuperscript{+} with $\epsilon=2$ under MTB attack on the CIFAR-10 dataset.} 
    \vspace{2mm}
    \begin{tabular}{c|ccccc}
    \hline
      $\delta_B$   & $q=1$ & $q=0.8$  & $q=0.6$ & $q=0.4$ & $q=0.2$\\
    \hline
     0\%  & 0.503  & \textbf{0.525} & 0.504 & 0.491 &  0.485 \\
     10\%    & 0.435 & 0.434 & \textbf{0.465} & 0.449 & 0.438\\
     20\%    & 0.255  & 0.284  & 0.297 & \textbf{0.328} & 0.241 \\
    \hline
    \end{tabular}
    \vspace{-3mm}
    \label{tab:vary_q}
\end{table}

\section{Related Work}
Due to limited space, we only discuss the most relevant defenses below and put other related work in Appendix \ref{apx:other_related_work}. Other works either only achieve DP or Byzantine robustness (but
not both), or combine secure aggregation with Byzantine robustness
without realizing DP.

Wang et al. \cite{wang2022privacy} proposed DDP-RP, an FL scheme offering Distributed DP (via encryption) and robustness (via range-proof technologies). However, this scheme only verifies if local model weights are within a bounded range, providing weak robustness. Our solution, in contrast, employs client momentum and centered clipping for Byzantine robustness with provable convergence. Zhu et al. \cite{zhu2022bridging} uses \emph{sign} aggregation for robustness, thus each client has limited impact, and adds DP noise to local gradients before sign operations. This method suffers from information loss, resulting in degraded convergence, and only accounts for the privacy cost of one iteration, underestimating the overall cost. Our solution, based on original SGD with momentum, considers the privacy cost of all iterations. Experimental results show that DP-BREM outperforms both approaches.

\section{Conclusions}
This paper aims to achieve FL in the cross-silo setting with both DP and Byzantine robustness. We first proposed DP-BREM, a DP version of LFH-based FL protocol with a robust aggregator based on client momentum, where the server adds noise to the aggregated momentum. Then we further developed DP-BREM\textsuperscript{+} which relaxes the server's trust assumption, by combining secure aggregation techniques with verifiable inputs and a new protocol for secure joint noise generation.  DP-BREM\textsuperscript{+} achieves the same DP and robustness guarantees as DP-BREM, under a malicious server (for privacy) and malicious minority clients. We theoretically analyze the error and convergence of DP-BREM, and conduct extensive experiments that empirically show the advantage of DP-BREM/\textsuperscript{+} in terms of privacy-utility tradeoff and  Byzantine robustness over five baseline protocols. In the future, we will extend our work to other types of robust aggregators.

\section*{Acknowledgments}
The authors would like to thank the anonymous reviewers and the shepherd for their valuable comments and suggestions. Li Xiong was partly supported by NSF grants CNS-2124104, CNS-2125530, IIS-2302968, and NIH grants R01LM013712, R01ES033241.

\bibliographystyle{plain}
\bibliography{mybibfile}

\appendix

\section{Proof of Theorem \ref{thm:privacy_analysis} (Privacy Analysis)}
\label{apx:proof_of_thm_privacy_analysis}
\begin{proof}
Since the added Gaussian noise in \eqref{equ:server_aggregate} has standard deviation $R\sigma$, and the aggregation sensitivity is shown in \eqref{equ:aggregation_sensitivity}, then the noise multiplier (defined by the ratio between Gaussian noise's standard deviation and the sensitivity) is 
\begin{align*}
    \sigma_i = \frac{R\sigma}{S_i}=\max\left\{\frac{R\sigma}{2C},\sigma p_i|\mathcal{D}_{i}|\right\}=\sigma\cdot\max\left\{\frac{R}{2C},p_i|\mathcal{D}_{i}|\right\}
\end{align*}
Also, due to the client-level sampling (i.e., each client was selected by the server w.p. $q$) and record-level sampling (i.e., each record was selected by client $\mathsf{C}_i$ w.p. $p_i$), the overall sampling rate is $qp_i$. Then, by applying the privacy accountant of Gaussian DP \cite{dong2019gaussian} (Lemma \ref{lem:GDP_privacy_account} in Appendix \ref{apx:GDP}), DP-BREM satisfies $\mu_i$-GDP with $\mu_i$ shown in \eqref{equ:mu_i}. Finally, by converting $\mu_i$-GDP to $(\epsilon_i, \delta)$-DP via Lemma \ref{lem:GDP_to_DP}, we get \eqref{equ:delta_epsilon}, which finishes the proof.
\end{proof}
\textbf{Remark: privacy accountant in practice.} 
Eq. \eqref{equ:delta_epsilon} provides the formula of $\delta$ when $\epsilon_i$ is given and $\mu_i$ is computed from \eqref{equ:mu_i}. In practice, however, we need to compute the value of privacy budget $\epsilon_i$ with a fixed $\delta$, where  $\delta$ is conventionally set to be less than $1/n$. In our experiments, we utilize the computation tool\footnote{https://github.com/woodyx218/Deep-Learning-with-GDP-Pytorch} in \cite{bu2020deep} to solve $\epsilon_i$ from \eqref{equ:delta_epsilon}. For the value of $\sigma_i$ in \eqref{equ:mu_i}, we usually have $p_i|\mathcal{D}_{i}|>\frac{R}{2C}$ in practice, then $\sigma_i=\sigma p_i|\mathcal{D}_{i}|$. In this case, the clipping bounds $R$ and $C$ are just hyperparameters that may affect the utility of the algorithm, but has no influence on the privacy analysis.

\section{Proof of Theorem \ref{thm:convergence_rate} (Convergence Rate)}
\label{apx:proof_of_thm_convergence_rate}
\begin{proof}
The proof of DP-BREM's convergence rate is based on the result of DP-BREM's aggregation error shown in Theorem \ref{thm:aggregation_error}, and LFH's convergence rate derived from LFH's aggregation error. Note that all differences between DP-BREM and LFH, including per-record clipping and the DP noise, are reflected by the aggregation error. Comparing with the aggregation error of $O(\rho^2|\mathcal{B}|/n)$ (ignoring constants and higher order terms) in LFH \cite[Lemma~9]{karimireddy2021learning}, our aggregation error shown in \eqref{equ:E_m_t_m_t_informal} replaces the term $|\mathcal{B}|$ by $|\mathcal{B}|+\sqrt{d}\sigma/q$, which means a slower convergence due to DP noise. Then, following the result in \cite[Theorem~VI]{karimireddy2021learning} and its informal version in \eqref{equ:convergence_rate_of_LFH}, we get the convergence rate of our algorithm as in  \eqref{equ:convergence_rate}. Note that our aggregation utilizes a client-level sampling rate $q$, i.e., approximate $nq$ clients participate in the aggregation for one iteration. We need to replace the term of $\frac{1}{n}$ in \eqref{equ:convergence_rate_of_LFH} by $\frac{1}{nq}$ in  \eqref{equ:convergence_rate}.
\end{proof}

\section{Detailed Steps of DP-BREM\textsuperscript{+} in Figure \ref{fig:framework_secure_aggregation}}
\label{apx:detailed_steps_of_secure_aggregation}

\textcircled{\small 1} Proof and Shares Generation: $\bm{z}_i, \mathsf{Valid}(\cdot)\rightarrow [\bm{z}_i]_j, [\pi_i]_j~(\forall j\neq i)$. For generating the proof, client $\mathsf{C}_i$ first evaluates the circuit $\mathsf{Valid}(\cdot)$ on its private input $\bm{z}_i$ to obtain the value of every wire in the arithmetic circuit corresponding to the computation of $\mathsf{Valid}(\bm{z}_i)$, then uses these wire values to generate the proof $\pi_i$ (refer to \cite{corrigan2017prio,roy2022eiffel} for the detailed format). Then, client $\mathsf{C}_i$ splits the private input $\bm{z}_i$ and proof $\pi_i$ to generate shares $[\bm{z}_i]_j$ and $ [\pi_i]_j~(\forall j\neq i)$, and send them to other clients $\{\mathsf{C}_j\}_{\forall j\neq i}$ via Shamir's secret sharing.

\textcircled{\small 2} Proof Summary Computation: $[\bm{z}_i]_j, [\pi_i]_j~(\forall j\neq i)\rightarrow[\sigma_i]_j~(\forall j\neq i)$. Each client except $\mathsf{C}_i$ first verifies the validity of the received secret shares via verifiable secret shares \cite{feldman1987practical}, and then locally constructs the shares of every wire in $\mathsf{Valid}(\bm{z}_i)$ via affine operations on the shares $[\bm{z}_i]_j$ and$ [\pi_i]_j$ to get the shares of proof summary $[\sigma_i]_j$ (refer to \cite{roy2022eiffel} for the detailed format), which will be sent to the server.

\textcircled{\small 3} Proof Summary Verification: $[\sigma_i]_j~(\forall j\neq i)\rightarrow\mathsf{Valid}(\bm{z}_i)$. After receiving shares of proof summary $[\sigma_i]_j(\forall j\neq i)$ from clients $\{\mathsf{C}_j\}_{\forall j\neq i}$, the server recovers the value of $\sigma_i$ via robust reconstruction, which is resilient to incorrect shares submitted by the malicious clients, and then checks the values in proof summaries. Finally, the validation result $\mathsf{Valid}(\bm{z}_i)=1$ if and only if $\sigma_i$ has the correct value.

\textcircled{\small 4} Random Numbers Generation: $l,d\rightarrow \{([u_k]_j,[v_k]_j)\}_{k=1}^{\lceil d/2 \rceil}~(\forall j)$. In this step, clients jointly generate the shares of $\lceil d/2 \rceil$-pairs of random numbers $\{(u_k,v_k)\}_{k=1}^{\lceil d/2 \rceil}$, where all of them are i.i.d. from uniform distribution in the range $[0,1]$. Denote $l$ as the fractional precision of the power 2 ring representation of real numbers. To obtain the share of one random number $u$, each client $\mathsf{C}_i~(\forall i)$ generates $l$ random bits in the binary filed $\mathbb{F}_2$, denoted by a binary vector $\bm{b}_i$ with length $l$, then generate and distributes the shares $[\bm{b}_i]_j$ to other clients (via Shamir's secret sharing). After receiving all shares from other clients, each client $\mathsf{C}_j~(\forall j)$ locally adds these shares to get $[\bm{b}]_j=[\sum_i\bm{b}_i]_j\in\mathbb{F}_2^l$, where vector $\bm{b}\in\mathbb{F}_2^l$ is actually the bitwise XOR of vectors $\{\bm{b}_i\}_{\forall i}$ because the computation is implemented in the binary field $\mathbb{F}_2^l$. We define the binary vector $\bm{b}$ as the binary representation of the fractional part of $u\in[0,1]$. Note that the Shamir's secret sharing scheme of Phase 1 is implemented in a finite filed $\mathbb{F}_{2^K}$, where $K>l$. Therefore, the client $\mathsf{C}_j$ can locally compute the arithmetic share $[u]_j\in\mathbb{F}_{2^K}$ from the share of binary representation $[\bm{b}]_j\in\mathbb{F}_{2}^l$. Since all possible discrete values with power 2 ring representation evenly span the range $[0,1]$, the generated random real number $u$ is uniformly distributed in $[0,1]$.

\textcircled{\small 5} Transformation to Gaussian Distribution: $\{([u_k]_j,[v_k]_j)\}_{k=1}^{\lceil d/2 \rceil}~(\forall j)\rightarrow [\bm{\xi}]_j~(\forall j)$. For each pair of $(u_k,v_k)$, clients can jointly compute a secret sharing of $a_k=\sqrt{-2\ln(u_k)}\cdot\cos(2\pi v_k)$ and of $b_k=\sqrt{-2\ln(u_k)}\cdot\sin(2\pi v_k)$ by utilizing Secure Multiparty Computation (MPC) protocols \cite{keller2020mp} that guarantees security (i.e., privacy and integrity) with malicious minority. According to Box and Muller Transformation \cite{box1958note}, $a_k$ and $b_k$ are i.i.d. random variables from the Gaussian distribution with mean 0 and variance 1. Then, by locally implementing secure multiplication with a constant (i.e., $R\sigma$),  $a_k$ and $b_k$ are i.i.d random numbers following a Gaussian distribution with the desired standard deviation of $R\sigma$. Finally, by concatenating shares of $d$ numbers in   $\{(a_k,b_k)\}_{k=1}^{\lceil d/2 \rceil}$, clients obtains the shares of random vector $\bm{\xi}$ with length $d$ from Gaussian distribution $\mathcal{N}(0,R^2\sigma^2 \mathbf{I}_d)$.

\textcircled{\small 6} Shares Aggregation: $\{[\bm{z}_i]_j\}_{i\in\mathcal{I}_\mathsf{Valid}}, [\bm{\xi}]_j~(\forall j)\rightarrow[\sum_{i\in\mathcal{I}_\mathsf{Valid}}\bm{z}_i+\bm{\xi}]_j~(\forall j)$. Due to the linearity of Shamir's secret sharing scheme, each client $\mathsf{C}_j$ can locally compute the share of the noisy aggregate by adding the shares of all valid inputs and the share of Gaussian noise: $[\sum_{i\in\mathcal{I}_\mathsf{Valid}}\bm{z}_i+\bm{\xi}]_j=\sum_{i\in\mathcal{I}_\mathsf{Valid}}[\bm{z}_i]_j+[\bm{\xi}]_j$, and sends that share to the server. 

\textcircled{\small 7} Noisy Aggregate Reconstruction: $[\sum_{i\in\mathcal{I}_\mathsf{Valid}}\bm{z}_i+\bm{\xi}]_j~(\forall j)\rightarrow\sum_{i\in\mathcal{I}_\mathsf{Valid}}\bm{z}_i+\bm{\xi}$. After receiving all shares of the noisy aggregate, the server recovers it using robust reconstruction. 

\section{Experimental Setup}
\label{apx:experimental_setup}

\textbf{FL Implementation.}
Due to limited resources, we simulate the distributed training of FL by running a single machine sequentially for clients and the server. The real-world implementation of FL is out of the scope of this paper.

\textbf{Datasets (non-IID) and Model Architecture.} We use three datasets for our experiments: MNIST \cite{lecun1998mnist} CIFAR-10 \cite{krizhevsky2009learning} and FEMNIST \cite{caldas2018leaf}, where the number of total clients is $n=100$ for the former two datasets, and $n=400$ for FEMNIST dataset. Note that the MNIST and CIFAR-10 datasets only have 10 classes, while the FEMNIST dataset has 62 classes (including 10 digits, 26 lowercase letters, and 26 uppercase letters). For the MNIST dataset, we use the CNN model from PyTorch example\footnote{\url{https://github.com/pytorch/opacus}}. For the CIFAR-10 dataset, we use the CNN model from the TensorFlow tutorial\footnote{\url{https://www.tensorflow.org/tutorials/images/cnn}}, like the previous works \cite{zheng2021federated,mcmahan2018learning}. To simulate the heterogeneous data distributions, we make non-i.i.d. partitions of the datasets, which is a similar setup as \cite{zheng2021federated} and is described below. For the FEMNIST dataset, we use a CNN model with 2 convolution layers and 2 fully connected layers. 

1) Non-IID MNIST: The MNIST dataset contains 60,000 training images and 10,000 testing images of 10 classes. There are 100 clients, each holds 600 training images. We sort the training data by digit label and evenly divide it into 400 shards. Each client is assigned four random shards of the data, so that most of the clients have examples of three or four digits.

2) Non-IID CIFAR-10: The CIFAR-10 dataset contains 50,000 training images and 10,000 test images of 10 classes. There are 100 clients, each holds 500 training images. We sample the training images for each client using a Dirichlet distribution with hyperparameter 0.9.

3) Non-IID FEMNIST: The FEMNIST dataset is pre-partitioned based on the writer of the characters, simulating a non-IID scenario. Each client's local dataset consists of samples written by individual users, introducing variability in handwriting styles. We use the TensorFlow-Federated API\footnote{\url{https://www.tensorflow.org/federated/api_docs/python/tff/simulation/datasets/emnist/load_data}} to load the first 400 partitions, representing data from 400 clients. Unlike the MNIST dataset, which includes only digits, FEMNIST includes both digits and uppercase and lowercase letters, spanning 62 classes (10 digits + 52 letters).

\textbf{Byzantine Attacks.}
We consider four different Byzantine attacks in our experiments.   

1) ALIE ("a little is enough") \cite{baruch2019little}. The attacker uses the empirical variance (estimated from the data of corrupted clients) to determine the perturbation range, in which the attack can deviate from the mean without being detected or filtered out.

2) IPM (inner-product manipulation) \cite{xie2020fall}. The attacker manipulates the submitted gradient to be the negative direction of the mean of other honest clients' gradients, thus the negative inner-product of the true gradient and the aggregation prevents the descent of the loss. Note that the original IPM attack assumes the \emph{omniscient} attacker (i.e., knows the data/gradient of all other clients), which is contradicted to our assumption that the attacker only has access to the data of the corrupted clients (otherwise, the privacy is already leaked and no need to provide DP). Thus, in the experiments, we use the data of corrupted clients to estimate the aggregated gradient of honest clients, and then manipulate the inner-product (i.e., non-omniscient attack).

3) LF (label-flipping). The attacker modifies the labels of all examples of corrupted clients' data and computes a new gradient, then uses a gradient replacement strategy (similar to \cite{bagdasaryan2020backdoor}) to enhance the impact on the global model. Specifically, the attacker computes a benign gradient $\mathbf{g}_\text{benign}$ with non-flipped labels and also a bad gradient $\mathbf{g}_\text{bad}$ with flipped labels. Finally, each malicious client submits the difference $\mathbf{g}_\text{bad}-\mathbf{g}_\text{benign}$ to the server to achieve the goal that the aggregated global gradient (averaged over all clients) is close to $\mathbf{g}_\text{bad}$.

4) MTB ("manipulating-the-Byzantine") \cite{shejwalkar2021manipulating}. The attacker computes a benign reference aggregate using some benign data samples obtained from corrupted clients, then computes a malicious perturbation vector, and an optimized scaling factor to get the malicious update with the goal of evading detection by robust aggregation algorithms. The optimization of the scaling factor can be tailored or agnostic to the aggregator. Considering our scheme and the baselines do not detect malicious clients, we use the agnostic setting (including min-max and min-sum) for simplicity because tailoring MTB attack to all defense aggregators is nontrivial. In our experiments, we implement the min-max attack since it has a larger impact on the global model.

\textbf{Byzantine Defenses with DP.}
We compare the performance of our approaches with the following five competitors against Byzantine attacks. All of them satisfy record-level DP via record-level clipping and DP noise added to the local gradient/momentum. Note that privacy budget $\epsilon$ in Theorem \ref{thm:privacy_analysis} is the same for different clients because clients have the same size of local datasets $|\mathcal{D}_i|$ and same record-level sampling rate (i.e., same $|\mathcal{D}_i|$ and $p_i$ for different clients $\mathsf{C}_i$). 

1) DP-FedSGD. Note that the original DP-FedSGD in \cite{mcmahan2018learning} clips the client gradient to achieve \emph{client-level} DP. For a fair comparison, we also implement record-level gradient clipping on top of the original DP-FedSGD to guarantee record-level DP. Though DP-FedSGD is not designed for robustness, its client-level clipping can restrict malicious clients' capability, thus providing some level of Byzantine robustness. We take this as a baseline to illustrate that client-level clipping can provide some level of robustness, but may not be enough to defend against strong attackers (either advanced attack strategy or a larger number of malicious clients).

2) DP-CM. As a baseline that adds DP to median-based robust aggregators (discussed in Section \ref{sec:challenges_and_baseline}), we implement the Byzantine-robust aggregator Coordinate-wise Median (CM) \cite{yin2018byzantine} with DP noise added to the median result. Note that only DP-CM uses median-based aggregation, while other methods use average-based aggregation. As discussed in Section \ref{sec:problem_statement} and Example \ref{exp:sensitivity_average_median}, the median-based aggregation has large sensitivity and poor privacy-utility tradeoff. 

3) DDP-RP \cite{wang2022privacy}. By leveraging encryption techniques, DDP-RP guarantees Distributed DP with secure aggregation. It allows clients to add smaller noise in the local gradient than the Local DP, with the knowledge of the lower bound of trusted clients, thus providing enhanced privacy-utility tradeoff than local DP protocols. To guarantee Byzantine robustness, DDP-RP uses range-proof (RP) technologies to securely verify whether the local model/gradient weights are in a (predefined) bounded range.

4) DP-RSA \cite{zhu2022bridging}. It replaces the \emph{value} aggregation to \emph{sign} aggregation, which  provides robustness because each client has limited impact on the aggregation. The DP noise is added to the local gradient before the sign operation.

5) DP-LFH.  The baseline (Section \ref{sec:challenges_and_baseline}) directly combines DP-SGD based momentum with LFH. Each client adds DP noise to the local gradient, and then computes the local momentum to be aggregated with centered clipping by the server. 

\section{Parameters in Experiments}
\label{apx:parameters_in_experiments}

\textbf{Basic Parameters.}
\begin{itemize}[leftmargin=7pt]
    \item Total number of iterations $T$: 1000 for MNIST and FEMNIST; 2000 for CIFAR-10.
    \item Learning rate $\eta_t$: For MNIST and FEMNIST datasets, $\eta_t$ is linearly reduced from 0.1 to 0.01 w.r.t. iterations. For CIFAR-10 dataset, $\eta_t$ is linearly reduced from 0.05 to 0.0025 w.r.t. iterations.
\end{itemize}

\textbf{DP-related Parameters.}
\begin{itemize}[leftmargin=7pt]
    \item Record-level sampling rate $p_i$: 0.05 for all $i$ on MNIST and CIFAR-10; 0.1 for all $i$ on FEMNIST (because each client has fewer data records).
    \item Client-level sampling rate $q$: the default value is 1. We evaluate the influence of $q$ (from 0.2 to 1) on the accuracy in Table \ref{tab:vary_q}.
    \item Record-level clipping bound $R$: linearly reduced from $R_0$ to $0.3R_0$ w.r.t. iterations. Note that in Figure \ref{fig:mnist_vary_R}, the different value of $R$ in x-axis is the value of the above $R_0$. For MNIST and FEMNIST, we set $R_0=10$ by default.  For CIFAR-10, we set $R_0=20$ by default, but $R_0=15$ only for the case of 
    $\epsilon=2$ in Figure \ref{fig:vary_epsilon}.
    \item Privacy parameter $\delta$ in DP: $10^{-6}$
    \item Noise multiplier $\sigma$: For MNIST (with $T=1000$ and each client has $|\mathcal{D}|_i=60000/100=600$ examples), $\sigma\in\{0.15, 0.06, 0.029, 0\}$ for $\epsilon\in\{1,3,8,\text{inf}\}$. For CIFAR-10 (with $T=2000$ and each client has $|\mathcal{D}|_i=50000/100=500$ examples), $\sigma\in\{0.14, 0.077, 0.042, 0\}$  for $\epsilon\in\{2,4,9,\text{inf}\}$. For FEMNIST (with $T=1000$ and each client has around $300$ examples), $\sigma\in\{0.16, 0.09, 0.047, 0\}$  for $\epsilon\in\{2,4,9,\text{inf}\}$.
\end{itemize}

\textbf{Robustness-related Parameters.}
\begin{itemize}[leftmargin=7pt]
    \item Client-level clipping bound $C$ (only for DP-BREM and DP-LFH): linearly reduced from $C_0$ to $0.3C_0$ w.r.t. iterations, where $C_0=1$ for MNIST, $C_0=5$ for CIFAR-10, and $C_0=2$ for FEMNIST.
    \item Momentum parameter $\beta=0.9$ for all three datasets.
\end{itemize}

\section{Proof of Theorem \ref{thm:aggregation_error} (Aggregation Error)}
\label{apx:proof_of_thm_aggregation_error}

Before proving Theorem \ref{thm:aggregation_error}, we first show some notations and assumptions. In $t$-th iteration, denote the selected honest clients $\mathcal{H}_t=\mathcal{H}\cup \mathcal{I}_t$ and selected Byzantine clients $\mathcal{B}_t=\mathcal{B}\cup \mathcal{I}_t$. For momentum updates in $t$-th iteration, we simplify the following notation (ignoring the subscript $t$) for convenience,
\begin{align*}
    \bm{y}_0\coloneqq \tilde{\bm{m}}_{t-1},\quad
    \bm{y}_i\coloneqq \bm{y}_0+\bm{z}_i \quad \text{with } \bm{z}_i\coloneqq \mathsf{Clip}_C(\bar{\bm{m}}_{t,i} - \bm{y}_0)
\end{align*}
where $\bar{\bm{m}}_{t,i}$ is the client momentum computed from gradient with record-level clipping. Then, we can rewrite the noisy global momentum as $\tilde{\bm{m}}_{t}=\frac{\sum_{i\in\mathcal{I}_t} \bm{y}_i +\bm{\xi}}{|\mathcal{I}_t|}$, where $\bm{\xi}\sim\mathcal{N}(0,R^2\sigma^2)$.

We assume $\{\bm{m}_{t,i}\}_{i\in\mathcal{H}}$ are i.i.d. with expectation $\bm{\mu}\coloneqq\mathbb{E}[\bm{m}_{t,i}]$ and variance  is bounded (in terms of L2-norm) $\mathbb{E}\|\bm{m}_{t,i}-\bm{\mu}\|^2\leqslant\rho^2$. Therefore, the record-level gradient clipped ones $\{\bar{\bm{m}}_{t,i}\}_{i\in\mathcal{H}}$ are also i.i.d., and we denote the expectation $\bar{\bm{\mu}}\coloneqq\mathbb{E}[\bar{\bm{m}}_{t,i}]$. Due to the clipping operation, the variance is reduced, and we assume $\mathbb{E}\|\bar{\bm{m}}_{t,i}-\bar{\bm{\mu}}\|^2\leqslant[\rho R/(R+c)]^2$, where $R$ is the record-level clipping bound and $c$ is some positive constant. Also, there is a gap between $\bm{\mu}$ and $\bar{\bm{\mu}}$ and we assume $\|\bar{\bm{\mu}}-\bm{\mu}\|^2\leqslant(\kappa/R)^2$. We assume $\bm{y}_0$ is not very far away from both $\bm{\mu}$ and $\bar{\bm{\mu}}$: $\|\bm{y}_0-\bar{\bm{\mu}}\|^2\leqslant\phi^2$ and $\|\bm{y}_0-\bm{\mu}\|^2\leqslant\tau^2$.

\begin{proof}
Our proof heavily relies on several useful lemmas shown in Appendix \ref{apx:useful_lemmas}, where Lemma \ref{lem:useful_inequality} splits the L2-norm of summation of vectors into \emph{weighted} summation of vectors'L2-norm, and Lemma \ref{lem:KKT_optimal_solution} provides the optimal strategy to choose these weights. 

We first consider the bound of $\mathbb{E}\|\tilde{\bm{m}}_{t}-\bm{\mu}\|^2$. Recall that the selected client set is $\mathcal{I}_t=\mathcal{H}_t\cup\mathcal{B}_t$, where the honest clients set $\mathcal{H}_t$ and Byzantine clients set $\mathcal{B}_t$ are disjoint. For any positive values $\gamma_1, \gamma_2, \gamma_3>0$ with $\gamma_1 + \gamma_2 + \gamma_3 = 1$, we have
{\small
\begin{align*}
    &\qquad|\mathcal{I}_t|^2\cdot\mathbb{E}\|\tilde{\bm{m}}_{t}-\bm{\mu}\|^2
     = |\mathcal{I}_t|^2\cdot\mathbb{E}\left\|\frac{\left(\sum_{i\in\mathcal{I}_t} \bm{y}_i\right) +\bm{\xi}}{|\mathcal{I}_t|}-\bm{\mu}\right\|^2 \\
    &\overset{(\mathsf{a})}{=} \mathbb{E}\left\| \sum\nolimits_{i\in\mathcal{H}_t}(\bm{y}_i-\bm{\mu}) + \sum\nolimits_{j\in\mathcal{B}_t}(\bm{y}_j-\bm{\mu}) + \bm{\xi}\right\|^2 \\
    &\overset{(\mathsf{b})}{\leqslant} \frac{1}{\gamma_1}\underbrace{\mathbb{E}\left\|\sum\nolimits_{i\in\mathcal{H}_t}(\bm{y}_i-\bm{\mu})\right\|^2}_{\mathcal{T}_1}
    + \frac{1}{\gamma_2}\underbrace{\mathbb{E}\left\|\sum\nolimits_{j\in\mathcal{B}_t}(\bm{y}_j-\bm{\mu})\right\|^2}_{\mathcal{T}_2}
    + \frac{1}{\gamma_3}\underbrace{\mathbb{E}\|\bm{\xi}\|^2}_{\mathcal{T}_3}
\end{align*} }
where $(\mathsf{a})$ used the fact that $\mathcal{I}_t=\mathcal{H}_t\cup\mathcal{B}_t$ and $\mathcal{H}_t\cap\mathcal{B}_t=\emptyset$; $(\mathsf{b})$ used the result in Lemma \ref{lem:useful_inequality}. From the above inequality, the error can be decomposed into three terms: $\mathcal{T}_1$ corresponds to the error of honest clients (who follow the protocol honestly) due to the randomness of clients' training data and bias introduced by clipping,  $\mathcal{T}_2$ corresponds to the error of Byzantine clients (who submit arbitrary $\bar{\bm{m}}_{t,i}$ but will be clipped by the server), and $\mathcal{T}_3$ corresponds to the error introduced by added Gaussian noise for privacy purpose. We will analyze each of the three errors in turn.

\textbf{Bounding $\mathcal{T}_1$.} 
Since $\mathbb{E}\|X\|^2=\|\mathbb{E}[X]\|^2 + \mathbb{E}\|X-\mathbb{E}[X]\|^2$ for any random vector $X$,  we can rewrite $\mathcal{T}_1$ as
{\small
\begin{align*}
    \mathcal{T}_1=\underbrace{{\left\|\sum\nolimits_{i\in\mathcal{H}_t}(\mathbb{E}[\bm{y}_i]-\bm{\mu})\right\|^2}}_{\mathcal{T}_{11}} + \underbrace{{\mathbb{E}\left\|\sum\nolimits_{i\in\mathcal{H}_t}(\bm{y}_i-\mathbb{E}[\bm{y}_i])\right\|^2}}_{\mathcal{T}_{12}}
\end{align*}}
where $\mathcal{T}_{11}$ corresponds to the bias introduced by the clipping operations, and $\mathcal{T}_{12}$ is the variance of honest clients' submissions. Rewrite $\bm{z}_i=\alpha_i\cdot(\bar{\bm{m}}_{t,i} - \bm{y}_0)$, where $\alpha_i=\min\{1, \frac{C}{\|\bar{\bm{m}}_{t,i} - \bm{y}_0\|}\}\in(0,1]$. Let $\mathbbm{1}_i$ be an indicator variable denoting if the momentum difference $\bar{\bm{m}}_{t,i} - \bm{y}_0$ was  clipped. Therefore, if $\|\bar{\bm{m}}_{t,i} - \bm{y}_0\|\leqslant C$, then $\mathbbm{1}_i=0$ and  $\alpha_i=1$; if $\|\bar{\bm{m}}_{t,i} - \bm{y}_0\|>C$, then $\mathbbm{1}_i=1$ and  $0<\alpha_i<1$.  Then, for each $i\in\mathcal{H}_t$, we have
{\small
\begin{align*}
    &\quad\mathbb{E}\|\bm{z}_i-(\bar{\bm{m}}_{t,i} - \bm{y}_0)\|
    = \mathbb{E}[(1-\alpha_i)\cdot\|\bar{\bm{m}}_{t,i} - \bm{y}_0\|]  \\
    &\leqslant \mathbb{E}[\mathbbm{1}_i\cdot\|\bar{\bm{m}}_{t,i} - \bm{y}_0\|] 
    \leqslant \frac{\mathbb{E}[\mathbbm{1}_i\cdot\|\bar{\bm{m}}_{t,i} - \bm{y}_0\|^2]}{C}
    \leqslant \frac{\mathbb{E}\|\bar{\bm{m}}_{t,i} - \bm{y}_0\|^2}{C}
\end{align*}}
where
{\small
\begin{align*}
    &\qquad\mathbb{E}\|\bar{\bm{m}}_{t,i} - \bm{y}_0\|^2 
    =\mathbb{E}\|(\bar{\bm{m}}_{t,i}-\bar{\bm{\mu}})+(\bar{\bm{\mu}}-\bm{y}_0)\|^2 \\
    &\overset{(\mathsf{a})}{\leqslant} \frac{\mathbb{E}\|\bar{\bm{m}}_{t,i}-\bar{\bm{\mu}}\|^2}{\gamma} + \frac{\mathbb{E}\|\bar{\bm{\mu}}-\bm{y}_0\|^2}{1-\gamma} 
    \leqslant \frac{[\rho R/(R+c)]^2}{\gamma} + \frac{\phi^2}{1-\gamma}\\
    &
    \overset{(\mathsf{b})}{=}[\rho R/(R+c)+\phi]^2
\end{align*}}
where $(\mathsf{a})$ is obtained by using Lemma \ref{lem:useful_inequality} for any $\gamma\in(0,1)$; $(\mathsf{b})$ is obtained by taking $\gamma=\frac{\rho R/(R+c)}{\rho R/(R+c)+\phi}$. Therefore,
{\small
\begin{align*}
    &\qquad\|\mathbb{E}[\bm{y}_i]-\bm{\mu}\|^2
    \overset{(\mathsf{a})}{=}\|\mathbb{E}[\bm{y}_0+\bm{z}_i-\bar{\bm{m}}_{t,i}] + (\bar{\bm{\mu}}-\bm{\mu})\|^2 \\
    &\overset{(\mathsf{b})}{\leqslant} \frac{\|\mathbb{E}[\bm{z}_i-(\bar{\bm{m}}_{t,i} - \bm{y}_0)]\|^2}{\gamma} + \frac{\|\bar{\bm{\mu}}-\bm{\mu}\|^2}{1-\gamma} \\
    &\overset{(\mathsf{c})}{\leqslant}\frac{\left(\mathbb{E}\|\bm{z}_i-(\bar{\bm{m}}_{t,i}-\bm{y}_0)\|\right)^2}{\gamma} + \frac{\|\bar{\bm{\mu}}-\bm{\mu}\|^2}{1-\gamma} \\
    &\overset{(\mathsf{d})}{\leqslant}\frac{[\rho R/(R+c)+\phi]^4}{\gamma C^2} + \frac{(\kappa/R)^2}{1-\gamma} \\
    &\overset{(\mathsf{e})}{=}\left[\frac{[\rho R/(R+c)+\phi]^2}{C} + \kappa/R\right]^2 
\end{align*}}
where $(\mathsf{a})$ is obtained from the definitions $\bm{y}_i=\bm{y}_0+\bm{z}_i$ and $\mathbb{E}[\bar{\bm{m}}_{t,i}]=\bar{\bm{\mu}}$; $(\mathsf{b})$ is obtained by using Lemma \ref{lem:useful_inequality} for any $\gamma\in(0,1)$; $(\mathsf{c})$ is derived from Jensen's Inequality, i.e., $\mathbb{E}[f(X)]\geqslant f(\mathbb{E}[X])$ for convex function $f(X)\coloneqq\|X\|$; $(\mathsf{d})$ is obtained by plugging in the previous two inequalities; $(\mathsf{e})$ is obtained by taking $\gamma=\frac{[\rho R/(R+c)+\phi]^2}{[\rho R/(R+c)+\phi]^2+C\kappa/R}$. Now, we can bound $\mathcal{T}_{11}$ by:
{\small
\begin{align*}
    \mathcal{T}_{11}
    \leqslant |\mathcal{H}_t|\sum_{i\in\mathcal{H}_t} \left\|\mathbb{E}[\bm{y}_i]-\bm{\mu}\right\|^2 
    \leqslant |\mathcal{H}_t|^2\left[\frac{[\rho R/(R+c)+\phi]^2}{C} + \frac{\kappa}{R}\right]^2
\end{align*}}
where the first inequality is obtained by using Lemma \ref{lem:useful_inequality}. On the other hand, we can bound $\mathcal{T}_{12}$ by
{\small
\begin{align*}
    \mathcal{T}_{12} 
    &\overset{(\mathsf{a})}{=} \mathbb{E}\sum\nolimits_{i\in\mathcal{H}_t}\|\bm{y}_i-\mathbb{E}[\bm{y}_i]\|^2
    \overset{(\mathsf{b})}{\leqslant}
    \mathbb{E}\sum\nolimits_{i\in\mathcal{H}_t}\|\bar{\bm{m}}_{t,i}-\mathbb{E}[\bar{\bm{m}}_{t,i}]\|^2 \\
    &\leqslant |\mathcal{H}_t|\cdot[\rho R/(R+c)+\phi]^2
\end{align*}}
where $(\mathsf{a})$ used the assumption that
$\{\bar{\bm{m}}_{t,i}\}_{i\in\mathcal{H}_t}$ are independent, then the random variables $\{\bm{y}_i\}_{i\in\mathcal{H}_t}$ are also independent; $(\mathsf{b})$ used contractivity of a clipping (projection) step. Thus,
{\small
\begin{align*}
    \mathcal{T}_{1}&=\mathcal{T}_{11} +\mathcal{T}_{12} \leqslant |\mathcal{H}_t|^2\left(\frac{\psi^2}{C} + \frac{\kappa}{R+c}\right)^2 + |\mathcal{H}_t|\psi^2\\
    &\leqslant \frac{4|\mathcal{H}_t|^2\psi^4}{C^2}+|\mathcal{H}_t|\psi^2 
    \leqslant \left(\frac{2|\mathcal{H}_t|\psi^2}{C}+\sqrt{|\mathcal{H}_t|}\psi\right)^2
\end{align*}}
where $\psi\coloneqq\rho R/(R+c)+\phi$, and the second inequality holds with the assumption
$C\leqslant\psi^2R/\kappa$ (thus we have $\frac{\psi^2}{C}\geqslant\frac{\kappa}{R}$)

\textbf{Bounding $\mathcal{T}_2$.}
For any Byzantine client $\mathsf{C}_j$ with $j\in\mathcal{B}_t$, the error is bounded by the clipping step
{\small
\begin{align*}
    \mathbb{E}\|\bm{y}_j-\bm{\mu}\|^2
    &= \mathbb{E}\|\bm{z}_j + (\bm{y}_0-\bm{\mu})\|^2 
    \overset{(\mathsf{a})}{\leqslant}
    \frac{\mathbb{E}\|\bm{z}_j\|^2}{\gamma} + \frac{\mathbb{E}\|\bm{y}_0-\bm{\mu}\|^2}{1-\gamma}\\
    &\overset{(\mathsf{b})}{\leqslant} 
    \frac{C^2}{\gamma} + \frac{\tau^2}{1-\gamma}
    \overset{(\mathsf{c})}{=}(C+\tau)^2
\end{align*}}
where $(\mathsf{a})$ is obtained by using in Lemma \ref{lem:useful_inequality} for any $\gamma\in(0,1)$; $(\mathsf{b})$ is obtained by the definition of $\bm{z}_j$ and the assumption; $(\mathsf{c})$ is obtained by taking $\gamma=\frac{C}{C+\tau}$. Then, by using Lemma \ref{lem:useful_inequality}, we have
{\small
\begin{align*}
    \mathcal{T}_2
    \leqslant
    |\mathcal{B}_t|\cdot\sum_{j\in\mathcal{B}_t} \mathbb{E}\|\bm{y}_j-\bm{\mu}\|^2
    \leqslant |\mathcal{B}_t|^2(C+\tau)^2
\end{align*}}

\textbf{Bounding $\mathcal{T}_3$.} 
Since the random noise  $\bm{\xi}\sim\mathcal{N}(0,R^2\sigma^2 \mathbf{I}_d)\in\mathbb{R}^d$, we have $\mathcal{T}_3=dR^2\sigma^2$.

\textbf{Putting It All Together.}
Combining all terms, we have
{\small
\begin{align}
    \label{equ:E_m_t_mu}
    &\qquad\mathbb{E}\|\tilde{\bm{m}}_{t}-\bm{\mu}\|^2 \notag\\
    &\leqslant\frac{1}{|\mathcal{I}_t|^2}\left[\frac{1}{\gamma_1}\left(\frac{2|\mathcal{H}_t|\psi^2}{C}+\sqrt{|\mathcal{H}_t|}\psi\right)^2
    + \frac{|\mathcal{B}_t|^2(C+\tau)^2}{\gamma_2} +\frac{dR^2\sigma^2}{\gamma_3}\right] \notag\\
    &\overset{(\mathsf{a})}{=}\frac{1}{|\mathcal{I}_t|^2}\left[\left(\frac{2|\mathcal{H}_t|\psi^2}{C}+\sqrt{|\mathcal{H}_t|}\psi\right)
    + |\mathcal{B}_t|(C+\tau) +\sqrt{d}R\sigma\right]^2 \notag\\
    &\overset{(\mathsf{b})}{\leqslant}\frac{1}{|\mathcal{I}_t|^2}\left[\frac{2\kappa|\mathcal{H}_t|(\rho+\phi)^2}{\phi^2\cdot R} + (|\mathcal{B}_t|+\sqrt{d}\sigma)R + \sqrt{|\mathcal{H}_t|}(\rho+\phi) + |\mathcal{B}_t|\tau\right]^2 \notag\\
    &\overset{(\mathsf{c})}{=}\frac{1}{|\mathcal{I}_t|^2}\left[\frac{2(\rho+\phi)}{\phi}\sqrt{2\kappa|\mathcal{H}_t|(|\mathcal{B}_t|+\sqrt{d}\sigma)}
    + \sqrt{|\mathcal{H}_t|}(\rho+\phi) + |\mathcal{B}_t|\tau\right]^2 \notag\\
    &=\left[\underbrace{\frac{2(\rho+\phi)}{|\mathcal{I}_t|\phi}\sqrt{2\kappa|\mathcal{H}_t|(|\mathcal{B}_t|+\sqrt{d}\sigma)}
    + \frac{\sqrt{|\mathcal{H}_t|}(\rho+\phi)}{|\mathcal{I}_t|} + \frac{|\mathcal{B}_t|\tau}{|\mathcal{I}_t|}}_{\eqqcolon\Phi}\right]^2
\end{align}}
where $(\mathsf{a})$ is obtained by taking $\gamma_k=\frac{\sqrt{\Phi_k}}{\sqrt{\Phi_1}+\sqrt{\Phi_2}+\sqrt{\Phi_3}}$ for $k=1,2,3$, where $\Phi_1\coloneqq\left(\frac{2|\mathcal{H}_t|\psi^2}{C}+\sqrt{|\mathcal{H}_t|}\psi\right)^2, \Phi_2\coloneqq|\mathcal{B}_t|^2(C+\tau)^2, \Phi_3\coloneqq dR^2\sigma^2$; $(\mathsf{b})$ is obtained by considering $\psi=\rho R/(R+c)+\phi\leqslant(\rho+\phi)$ and taking the clipping bound $C=\frac{\phi^2}{\kappa}R$, which makes the previous assumption $C\leqslant\psi^2R/\kappa$ holds; $(\mathsf{c})$ is obtained by taking $R=\frac{\rho+\phi}{\phi}\sqrt{\frac{2\kappa|\mathcal{H}_t|}{|\mathcal{B}_t|+\sqrt{d}\sigma}}$, where $|\mathcal{H}_t|\approx|\mathcal{H}|q$ and $|\mathcal{B}_t|\approx|\mathcal{B}|q$. Since $|\mathcal{H}|+|\mathcal{B}|=n$ and $|\mathcal{B}|/n<1/2$, we can approximate the tuning by $R\propto  O(\rho\sqrt{n/(|\mathcal{B}_t|+\sqrt{d}\sigma/q})$.

\textbf{The Final Result.} On the other hand, we have
{\small
\begin{align}
    \label{equ:E_mu_m_t}
    \mathbb{E}\|\bm{\mu}-\bm{m}_t^*\|^2
    &=\frac{1}{|\mathcal{H}|^2}\mathbb{E}\left\|\sum\nolimits_{i\in\mathcal{H}}(\bm{m}_{t,i}-\bm{\mu})\right\|^2 \notag\\
    &=\frac{1}{|\mathcal{H}|^2}\mathbb{E}\sum\nolimits_{i\in\mathcal{H}}\left\|\bm{m}_{t,i}-\bm{\mu}\right\|^2 \leqslant\frac{\rho^2}{|\mathcal{H}|}
\end{align}}
where the first equality is obtained by the definition of $\bm{m}_t^*$; the second equality is obtained by the fact that all honest clients' momentum $\{\bm{m}_{t,i}\}_{i\in\mathcal{H}}$ are independent with each other; and the third  equality is obtained by the assumption $\left\|\bm{m}_{t,i}-\bm{\mu}\right\|^2\leqslant\rho^2$ for $i\in\mathcal{H}$. Finally, we have
{\small
\begin{align}
    \label{equ:E_m_t_m_t}
    &\qquad\mathbb{E}\|\tilde{\bm{m}}_{t}-\bm{m}_t^*\|^2  \notag\\
    &= \mathbb{E}\|(\tilde{\bm{m}}_{t}-\bm{\mu}) + (\bm{\mu}-\bm{m}_t^*)\|^2 
    \overset{(\mathsf{a})}{\leqslant}\frac{\mathbb{E}\|\tilde{\bm{m}}_{t}-\bm{\mu}\|^2}{\gamma} + \frac{\mathbb{E}\|\bm{\mu}-\bm{m}_t^*\|^2}{1-\gamma} \notag\\
    &\overset{(\mathsf{b})}{\leqslant} \frac{\Phi^2}{\gamma} + \frac{\rho^2/|\mathcal{H}|}{1-\gamma}
    \overset{(\mathsf{c})}{=} \left(\Phi+\frac{\rho}{\sqrt{|\mathcal{H}|}}\right)^2
\end{align}}
where $(\mathsf{a})$ is obtained by using Lemma \ref{lem:useful_inequality} for any $\gamma\in(0,1)$; $(\mathsf{b})$ is obtained from \eqref{equ:E_m_t_mu} and \eqref{equ:E_mu_m_t}, where $\Phi$ is defined in \eqref{equ:E_m_t_mu}; $(\mathsf{c})$ is obtained by taking $\gamma=\frac{\Phi}{\Phi+\frac{\rho}{\sqrt{|\mathcal{H}|}}}$. Furthermore, if we assume $\phi\leqslant O(\rho)$ and $\tau\leqslant O(\rho)$, we can rewrite \eqref{equ:E_m_t_m_t} as the following version 
\begin{align*}
    \mathbb{E}\|\tilde{\bm{m}}_{t}-\bm{m}_t^*\|^2 \leqslant O\left(\frac{\rho^2(|\mathcal{B}|+\sqrt{d}\sigma/q)}{n}\right)
\end{align*}
which finishes the proof of Theorem \ref{thm:aggregation_error}.
\end{proof}

\section{Useful Lemmas}
\label{apx:useful_lemmas}

\begin{lemma}
\label{lem:useful_inequality}
For any positive real values $\alpha_1,\cdots,\alpha_K\in\mathbb{R}^{+}$ and any $d$-dimensional vectors $\bm{x}_1,\cdots,\bm{x}_K\in\mathbb{R}^d$ with L2-norm $\|\cdot\|$, we have $\left\|\sum\nolimits_{k=1}^K\bm{x}_k\right\|^2
    \leqslant
    \left(\sum\nolimits_{k=1}^K \alpha_k\right) \cdot\left(\sum\nolimits_{k=1}^K \frac{\|\bm{x}_k\|^2}{\alpha_k}\right)$.
\end{lemma}

\begin{proof}
Denote $x_{ki}$ as the $i$-th element of the vector $\bm{x}_k$, then
{\small
\begin{align*}
    &\quad\left\|\sum\nolimits_{k=1}^K\bm{x}_k\right\|^2 
    = \sum\nolimits_{i=1}^d \left(\sum\nolimits_{k=1}^K x_{ki}\right)^2
    = \sum\nolimits_{i=1}^d  \left(\sum\nolimits_{k=1}^K \sqrt{\alpha_k}\cdot \frac{x_{ki}}{\sqrt{\alpha_k}}\right)^2 \\
    &\leqslant \sum\nolimits_{i=1}^d \left[\sum\nolimits_{k=1}^K \left(\sqrt{\alpha_k}\right)^2\cdot \sum\nolimits_{k=1}^K\left(\frac{x_{ki}}{\sqrt{\alpha_k}}\right)^2\right] \\
    &=\left(\sum\nolimits_{k=1}^K \alpha_k\right) \cdot\left(\sum\nolimits_{k=1}^K \frac{1}{\alpha_k}\sum\nolimits_{i=1}^d x_{ki}^2\right)
    =\left(\sum\nolimits_{k=1}^K \alpha_k\right) \cdot\left(\sum\nolimits_{k=1}^K \frac{\|\bm{x}_k\|^2}{\alpha_k}\right)
\end{align*} }
where the inequality is caused by Cauchy-Schwarz inequality.
\end{proof}

\begin{lemma}
\label{lem:KKT_optimal_solution}
Consider the following optimization problem 
\begin{align*}
    f^* = \min_{x_1,\cdots,x_K}~~\sum\nolimits_{k=1}^K c_k / x_k,\qquad
    \text{s.t.}\quad x_k>0,\quad \sum\nolimits_{k=1}^K x_k=1
\end{align*}
where $c_1,\cdots,c_K>0$.
Then, we have $f^*=(\sum\nolimits_{j=1}^K \sqrt{c_j})^2$, where the optimal solution is $x_k=\frac{\sqrt{c_k}}{\sum\nolimits_{j=1}^K \sqrt{c_j}}~(\forall k=1,\cdots,K)$.
\end{lemma}
\begin{proof}
The augmented Lagrange function is $\mathcal{L}(x_k;\lambda)=\sum\nolimits_{k=1}^K \frac{c_k}{x_k}+\lambda\cdot(\sum\nolimits_{k=1}^K x_k-1)$. By taking Karush-Kuhn-Tucker (KKT) conditions, we have
\begin{align*}
    \begin{cases}
    \frac{\partial \mathcal{L}}{\partial x_k}=0 \\
    \frac{\partial \mathcal{L}}{\partial \lambda}=0
    \end{cases}
    \Rightarrow~~
    \begin{cases}
    -\frac{c_k}{x_k^2}+\lambda=0 \\
    \sum\nolimits_{k=1}^K x_k=1
    \end{cases}
    \Rightarrow~~
    \begin{cases}
    x_k=\sqrt{\frac{c_k}{\lambda}} \\
    \sqrt{\lambda}=\sum\nolimits_{j=1}^K \sqrt{c_j}
    \end{cases}
\end{align*}
then we have $f^*=\left(\sum\nolimits_{j=1}^K \sqrt{c_j}\right)^2$, which finished the proof.
\end{proof}

\section{Proof of Lemma \ref{lem:aggregation_sensitivity} (Aggregation Sensitivity)}
\label{apx:proof_of_lem_aggregation_sensitivity}

\begin{proof}
    For the local momentum computation in \eqref{equ:local_momentum}, we can rewrite it as
    \begin{align*}
        \bm{m}_{t,i}
        =(1-\beta)(\bar{\bm{g}}_{t,i}+\beta \bar{\bm{g}}_{t-1, i}+\cdots+\beta^{t-2}\bar{\bm{g}}_{2,i})+\beta^{t-1} \bar{\bm{g}}_{1,i}
    \end{align*}
    For a neighboring dataset $\mathcal{D}_i^\prime$ which differs only one record from client $\mathsf{C}_i$'s data, i.e., $|\mathcal{D}_i-\mathcal{D}_i^\prime|=1$, we denote the corresponding local gradient (with per-record gradient clipping) and momentum as $\bar{\bm{g}}_{t,i}^\prime$ and $\bm{m}_{t,i}^\prime$, respectively. Since $\bar{\bm{g}}_{\tau,i}$ is computed by \eqref{equ:local_gradient} for $\tau=1,\cdots,t$, we have
    \begin{align*}
        \|\bar{\bm{g}}_{\tau,i}-\bar{\bm{g}}_{\tau,i}^\prime\|
        =\frac{1}{p_i|\mathcal{D}_i|}\|\mathsf{Clip}_R(\nabla_{\bm{\theta}}\ell(\bm{x}, \bm{\theta}_{\tau-1}) )\|\leqslant\frac{R}{p_i|\mathcal{D}_i|}
    \end{align*}
    where $\bm{x}=\mathcal{D}_i-\mathcal{D}_i^\prime$. Then,
    \begin{align*}
        \|\bm{m}_{t,i}-\bm{m}_{t,i}^\prime\|
        &\leqslant (1-\beta)[\|\bar{\bm{g}}_{t,i}-\bar{\bm{g}}_{t,i}^\prime\|+\beta \|\bar{\bm{g}}_{t-1, i}-\bar{\bm{g}}_{t-1, i}^\prime\|+\\
        &\quad\cdots+\beta^{t-2}\|\bar{\bm{g}}_{2,i}-\bar{\bm{g}}_{2,i}^\prime\|]+\beta^{t-1} \|\bar{\bm{g}}_{1,i}-\bar{\bm{g}}_{1,i}^\prime\|  \\
        &\leqslant [(1-\beta)(1+\beta+\cdots+\beta^{t-2}) + \beta^{t-1}]\cdot\frac{R}{p_i|\mathcal{D}_{i}|}  \\
        &= \left[(1-\beta)\cdot\frac{1-\beta^{t-1}}{1-\beta} + \beta^{t-1}\right] \cdot\frac{R}{p_i|\mathcal{D}_{i}|} 
        =\frac{R}{p_i|\mathcal{D}_{i}|} 
    \end{align*}
    where the first inequality is obtained by generalizing the triangle inequality; Therefore,
    \begin{align*}
        &\quad\|Q_t(\mathcal{D})-Q_t(\mathcal{D^\prime})\| \\
        &= \|\sum\nolimits_{j\in\mathcal{I}_t}\mathsf{Clip}_C(\bm{m}_{t,j} - \tilde{\bm{m}}_{t-1}) - \sum\nolimits_{j\in\mathcal{I}_t}\mathsf{Clip}_C(\bm{m}_{t,j}^\prime - \tilde{\bm{m}}_{t-1})\| \\
        &\overset{(\mathsf{a})}{=}\|\mathsf{Clip}_C(\bm{m}_{t,i} - \tilde{\bm{m}}_{t-1}) - \mathsf{Clip}_C(\bm{m}_{t,i}^\prime - \tilde{\bm{m}}_{t-1})\| \\
        &\overset{(\mathsf{b})}{\leqslant} \min\{2C,\|\bm{m}_{t,i}-\bm{m}_{t,i}^\prime\|\}
        =\min\{2C,\frac{R}{p_i|\mathcal{D}_{i}|}\}
    \end{align*}
    where $(\mathsf{a})$ is obtained due to $\bm{m}_{t,j}=\bm{m}_{t,j}^\prime$ for $j\neq i$; and $(\mathsf{b})$ is obtained according to Lemma \ref{lem:clipping_max_bound}.
    Now we finished the main proof of Lemma \ref{lem:aggregation_sensitivity}.
\end{proof}

\begin{table}[!t]
\centering
\caption{Notations}
\footnotesize
\begin{tabular}{c|c}
\hline
\textbf{Symbols}     &  \textbf{Description}\\
\hline
 $\bm{\theta}_t$    &  (global) model in $t$-th iteration, where $\bm{\theta}_t\in\mathbb{R}^d$ \\
 $\mathcal{D}_i$    & local training data of client $\mathsf{C}_i$ \\
 $\mathcal{D}_{i,t}$    & data batch sampled from $\mathcal{D}_i$ in $t$-th iteration \\
 $\bm{g}_{t,i}, \bm{m}_{t,i}$ &  client gradient and momentum at $t$-th iteration \\
 $\bm{m}_{t}$ & aggregation of $\bm{m}_{t,i}$ among multiple clients \\
 $\bar{\bm{g}}_{t,i}, \bar{\bm{m}}_{t,i}$ &  client gradient and momentum with record-level clipping   \\
 $\tilde{\bm{m}}_{t}$ & aggregation of $\bar{\bm{m}}_{t,i}$ with DP noise \\
 $p_i$ & record-level sampling rate (implemented by client $\mathsf{C}_i$) \\
 $q$ & client-level sampling rate (implemented by the server) \\
 $R$ &  record-level clipping bound (for DP) \\
 $C$ &  client-level clipping bound (for robustness) \\
  $\mathcal{H}$  &  the set of honest clients that follow the protocol honestly \\
  $\mathcal{B}$  & the set of Byzantine clients that are malicious \\
  $\delta_B$ & the percentage of Byzantine clients, i.e., $\delta_B=|\mathcal{B}|/n\times 100\%$\\
\hline
\end{tabular}
\end{table}

The above proof used the following lemma.
\begin{lemma}
    \label{lem:clipping_max_bound}
    For any vectors $x$ and $\delta$,  we have
    \begin{align*}
        \|\mathsf{Clip}_C(x) - \mathsf{Clip}_C(x+\delta)\|\leqslant  \min\{2C,\|\delta\|\}
    \end{align*}
    where $\mathsf{Clip}_C(x)\coloneqq \min\{1,C/\|x\|\} \cdot x$ and $\|\cdot\|$ denotes L2-norm.
\end{lemma}
\begin{proof}
Our proof of Lemma \ref{lem:clipping_max_bound} mainly uses the triangle inequality of a norm. Note that for L-2 norm $\|\cdot\|$,  we have$\|a+b\|^2=\|a\|^2 + 2a^\top b+ \|b\|^2$ for any vectors $a$ and $b$. We first show that $\|\mathsf{Clip}_C(x) - \mathsf{Clip}_C(x+\delta)\|\leqslant \|\delta\|$, which is proved by enumerating all cases as follows.

\textbf{Case 1.} Assume $\|x\|\leqslant C$ and $\|x+\delta\|\leqslant C$. Then,
\begin{align*}
\|\mathsf{Clip}_C(x) - \mathsf{Clip}_C(x+\delta)\|=\|x-(x+\delta)\| = \|\delta\|
\end{align*}
 
\textbf{Case 2.} Assume $\|x\| > C$ and $\|x+\delta\|\leqslant C$. Then, $0<\frac{C}{\|x\|}<1$ and 
\begin{align*}
    &\quad\|\mathsf{Clip}_C(x) - \mathsf{Clip}_C(x+\delta)\|^2 \\
    &=\left\|\frac{C}{\|x\|}\cdot x-(x+\delta)\right\|^2
    = \left\|\left(1-\frac{C}{\|x\|}\right)\cdot x+\delta\right\|^2 \\
    &= \left(1-\frac{C}{\|x\|}\right)^2 \|x\|^2 + 2\left(1-\frac{C}{\|x\|}\right)x^\top\delta + \|\delta\|^2 \\
    &=\left(1-\frac{C}{\|x\|}\right) \left[\left(1-\frac{C}{\|x\|}\right)\|x\|^2 + 2x^\top\delta + \|\delta\|^2\right] + \frac{C\cdot\|\delta\|^2}{\|x\|} \\
    &=\left(1-\frac{C}{\|x\|}\right) \left[\|x+\delta\|^2
    -C\|x\|\right] + \frac{C\cdot\|\delta\|^2}{\|x\|} \\
    &< 0 + 1\cdot \|\delta\|^2 = \|\delta\|^2
\end{align*}
where $\|x+\delta\|^2\leqslant C^2 < C\|x\|$. Therefore, $\|\mathsf{Clip}_C(x) - \mathsf{Clip}_C(x+\delta)\| < \|\delta\|$ in this case.

\textbf{Case 3.} Assume $\|x\| \leqslant C$ and $\|x+\delta\|> C$. Then, $0<\frac{C}{\|x+\delta\|}<1$ and
{\small
\begin{align*}
    &\quad\|\mathsf{Clip}_C(x) - \mathsf{Clip}_C(x+\delta)\|^2\\
    &=\left\| x-\frac{C}{\|x+\delta\|}\cdot(x+\delta)\right\| ^2
    =\left\|\left(1-\frac{C}{\|x+\delta\|}\right)\cdot (x+\delta) - \delta\right\|^2 \\
    &= \left(1-\frac{C}{\|x+\delta\|}\right)^2\|x+\delta\|^2 - 2\left(1-\frac{C}{\|x+\delta\|}\right) (x+\delta)^\top \delta  + \|\delta\|^2 \\
    &=\left(1-\frac{C}{\|x+\delta\|}\right)\left[\left(1-\frac{C}{\|x+\delta\|}\right)\|x+\delta\|^2-2(x+\delta)^\top \delta + \|\delta\|^2\right] \\
    &\qquad + \frac{C\cdot\|\delta\|^2}{\|x+\delta\|} \\
    &=\left(1-\frac{C}{\|x+\delta\|}\right)\left[
    \|(x+\delta)-\delta\|^2-C\|x+\delta\|\right] + \frac{C\cdot\|\delta\|^2}{\|x+\delta\|} \\ 
    &< 0 + 1\cdot \|\delta\|^2 = \|\delta\|^2 
\end{align*} }
where $\|(x+\delta)-\delta\|^2=\|x\|^2\leqslant C^2 < C\|x+\delta\|$. Therefore, $\|\mathsf{Clip}_C(x) - \mathsf{Clip}_C(x+\delta)\| < \|\delta\|$ in this case.

\textbf{Case 4.} Assume $\|x\| > C$ and $\|x+\delta\|> C$. Then,
\begin{align*}
    &\quad\|\mathsf{Clip}_C(x) - \mathsf{Clip}_C(x+\delta)\|^2\\
    &=\left\| \frac{C}{\|x\|}\cdot x-\frac{C}{\|x+\delta\|}\cdot(x+\delta)\right\|^2 \\
    &=  \frac{C^2}{\|x\|^2}\cdot \|x\|^2 - \frac{C^2\cdot 2x^\top(x+\delta)}{\|x\|\cdot\|x+\delta\|} + \frac{C^2}{\|x+\delta\|^2}\cdot \|x+\delta\|^2 \\
    &= 2C^2 - \frac{C^2\cdot[\|x\|^2 + (\|x\|^2+2x^\top\delta+ \|\delta\|^2)]}{\|x\|\cdot\|x+\delta\|} + \frac{C^2\cdot\|\delta\|^2}{\|x\|\cdot\|x+\delta\|} \\
    &= 2C^2 - \frac{C^2\cdot[\|x\|^2 + \|x+\delta\|^2]}{\|x\|\cdot\|x+\delta\|} + \frac{C^2}{\|x\|\cdot\|x+\delta\|} \cdot\|\delta\|^2\\
    &< 2C^2 - C^2\cdot2 + 1\cdot \|\delta\|^2 = \|\delta\|^2
\end{align*}
where $\frac{\|x\|^2 + \|x+\delta\|^2}{\|x\|\cdot\|x+\delta\|}\geqslant 2$ due to $\|x\|^2 + \|x+\delta\|^2 - 2(\|x\|\cdot\|x+\delta\|)=(\|x\|-\|x+\delta\|)^2\geqslant0$, and $\frac{C^2}{\|x\|\cdot\|x+\delta\|}<1$ due to $\|x\| > C$ and $\|x+\delta\|> C$. Therefore, $\|\mathsf{Clip}_C(x) - \mathsf{Clip}_C(x+\delta)\| < \|\delta\|$ in this case.

\textbf{The Final Result.} By summarizing all cases above, we have $\|\mathsf{Clip}_C(x) - \mathsf{Clip}_C(x+\delta)\|\leqslant \|\delta\|$. On the other hand, since $\|\mathsf{Clip}_C(x)\|\leqslant C$ for any $x$, it is obvious that 
{\small
\begin{align*}
    \|\mathsf{Clip}_C(x) - \mathsf{Clip}_C(x+\delta)\|\leqslant   \|\mathsf{Clip}_C(x) \| + \|\mathsf{Clip}_C(x+\delta)\| \leqslant  2C
\end{align*} }
Thus, the upper bound of $\|\mathsf{Clip}_C(x) - \mathsf{Clip}_C(x+\delta)\|$ is $\min\{2C,\|\delta\|\}$.
\end{proof}

\section{Gaussian Differential Privacy (GDP)}
\label{apx:GDP}
\textbf{Privacy Accountant.} Since deep learning needs to iterate over the training data and apply gradient computation multiple times during the training process, each access to the training data incurs some privacy cost from the overall privacy budget $\epsilon$. The total privacy cost of repeated applications of additive noise mechanisms follow from the composition theorems and their refinements \cite{dwork2014algorithmic}. The task of keeping track of the accumulated privacy loss in the course of execution of a composite mechanism, and enforcing the applicable privacy policy, can be performed by the privacy accountant. Abadi et al. \cite{abadi2016deep} proposed \emph{moments accountant} to provide a tighter bound on the privacy loss compared to the generic advanced composition theorem \cite{dwork2010boosting}. Another new and more state-of-the-art privacy accountant method is Gaussian Differential Privacy (GDP) \cite{dong2019gaussian,bu2020deep}, which was shown to obtain a tighter result than moments accountant.

\textbf{Gaussian Differential Privacy.}
GDP is a new privacy notion which faithfully retains  hypothesis testing interpretation of differential privacy. By leveraging the central limit theorem of Gaussian distribution, GDP has been shown to possess an \emph{analytically tractable} privacy accountant (vs. moments accountant must be done by numerical computation). Furthermore, GDP can be converted to a collection of $(\epsilon,\delta)$-DP guarantees (refer to Lemma \ref{lem:GDP_to_DP}). Note that even in terms of $(\epsilon,\delta)$-DP, the GDP approach gives a tighter privacy accountant than moments accountant.  GDP utilizes a single parameter $\mu\geqslant0$ (called privacy parameter) to quantify the privacy of a randomized mechanism. Similar to the privacy budget $\epsilon$ defined in DP,  a larger  $\mu$ in GDP indicates less privacy guarantee. Comparing with $(\epsilon,\delta)$-DP,  the new notion $\mu$-GDP can losslessly reason about common primitives associated with differential privacy, including composition, privacy amplification by sampling, and group privacy. In the following, we briefly introduce some important properties (that will be used in the analysis of our approach) of GDP as below. The formal definition and more detailed results can be found in the original paper \cite{dong2019gaussian}. 

\begin{lemma}[Gaussian Mechanism for GDP \cite{dong2019gaussian}]
Consider the problem of privately releasing a univariate statistic $f(D)$ of a dataset $D$.  Define the sensitivity of $f(\cdot)$ as $s_f=\sup_{D,D^\prime}|f(D)-f(D^\prime)|$, where the supremum is over all neighboring datasets.
Then,  the Gaussian mechanism $\mathcal{M}(D)=f(D)+\xi$, where $\xi\sim\mathcal{N}(0,s_f^2/\mu^2)$, satisfies  $\mu$-GDP. 
\end{lemma}

\begin{lemma}[Composition Theorem of GDP \cite{dong2019gaussian}]
The $m$-fold composition of $\mu_i$-GDP mechanisms is $\sqrt{\mu_1^2+\cdots+\mu_m^2}$-GDP. 
\end{lemma}

\begin{lemma}[Group Privacy of GDP \cite{dong2019gaussian}]
	If a mechanism is $\mu$-GDP, then it is $K\mu$-GDP for a group with size $K$. 
\end{lemma}

\begin{lemma}[$\mu$-GDP to $(\epsilon,\delta)$-DP \cite{dong2019gaussian}]
	\label{lem:GDP_to_DP}
	A mechanism is $\mu$-GDP if and only if it is $(\epsilon,\delta(\epsilon))$-DP for all $\epsilon\geqslant0$, where
	\begin{align*}
		\delta(\epsilon)=\Phi\left(-\frac{\epsilon}{\mu}+\frac{\mu}{2}\right)-e^\epsilon\cdot\Phi\left(-\frac{\epsilon}{\mu}-\frac{\mu}{2}\right),
	\end{align*}
	and $\Phi$ denotes the CDF of standard normal (Gaussian) distribution. 
\end{lemma}

\begin{lemma}[Privacy Central Limit Theorem of GDP \cite{bu2020deep}]
\label{lem:GDP_privacy_account}

Denote $p$ as the sampling probability of one example in the training dataset, $T$ as the total number of iterations and $\sigma$ as the noise scale (i.e., the ratio between the standard deviation of Gaussian noise and the gradient norm bound). Then,
 algorithm DP-SDG asymptotically satisfies $\mu$-GDP with privacy parameter $\mu=p\sqrt{T(e^{1/\sigma^2}-1)}$. 
\end{lemma}

In this paper, we use $\mu$-GDP as our primary privacy accountant method due to its good property on composition and accountant of privacy amplification and then convert the $\mu$-GDP result to $(\epsilon,\delta)$-DP. We note that other privacy accountant methods, such as moments accountant \cite{abadi2016deep} and R{\'e}nyi DP (RDP) \cite{mironov2017renyi}, are also applicable to the proposed scheme and theoretical analysis, but might lead to suboptimal results.

\section{Preliminaries for Crypto Primitives}
\label{apx:secure_aggregation_preliminaries}

\textbf{Shamir's Secret Sharing with Robust Reconstruction.}
Due to the assumption of a malicious minority, the utilized crypto primitives should be able to tolerate the wrong or missing messages of malicious clients. Shamir's $t$-out-of-$n$ Secret Sharing Scheme \cite{shamir1979share} allows distributing a secret $s$ among $n$ parties such that: 1) the complete secret can be reconstructed from any combination of $t$ shares; 2) any set of $t-1$ or fewer shares reveals no information about $s$, where $t$ is the threshold of the secret sharing scheme. We denote $[s]_i$ as the share held by the $i$-th party. Shamir's secret sharing scheme is linear, which means a party can locally perform: 1) addition of two shares, 2) addition of a constant, and 3) multiplication by a constant. Furthermore, Shamir's secret sharing scheme is closely related to Reed-Solomon error correcting codes \cite{lin2001error}, which is a group of polynomial-based error correcting codes. Shamir's secret sharing scheme results in a $[n,t,n-t+1]$ Reed-Solomon code that can tolerate up to $q$ errors and $e$ erasures (message dropouts) such that $2q+e<n-t+1$. Given any subset of $n-e$ shares $\mathcal{Q}~(|\mathcal{Q}|\geqslant n-e)$ with up to $q$ errors, any standard Reed Solomon decoding algorithm, such as Gao's decoding algorithm \cite{gao2003new}, can robustly reconstruct the secret $s$. Due to the property of robust reconstruction, Shamir's secret sharing is able to guarantee security with malicious minority (as versus additive secret sharing \cite{cramer2005share} guarantees security with honest-but-curious parties).

\textbf{EIFFeL: An Instantiation of SAVI Protocol.} 
EIFFeL \cite{roy2022eiffel} is a SAVI protocol (with privacy and integrity guarantees) that securely aggregates only well-informed inputs. Its threat model assumes a malicious server (for privacy only) and a set of malicious clients (for both breaching privacy and submitting malformed inputs) that can arbitrarily deviate from the protocol, while the remaining honest clients are assumed to follow the protocol correctly and have well-formed inputs.  EIFFeL ensures privacy by using Shamir's secret sharing scheme \cite{shamir1979share}. Integrity is guaranteed via 1) secret-shared non-interactive proofs (SNIP) \cite{corrigan2017prio}, which is an information-theoretic zero-knowledge proof for secret-shared data; and 2) verifiable secret shares \cite{feldman1987practical}, which validates the correctness of the secret shares. Note that the original SNIP utilizes additive secret sharing scheme \cite{cramer2005share}, and its deployment setting uses $\geqslant2$ honest and non-colluding servers as the verifiers. In contrast, by leveraging Shamir's secret sharing with robust reconstruction, EIFFeL extends SNIP to a malicious threat model in a \emph{single} server setting, where all the other clients (some of them are malicious) and the server jointly act as the verifiers for the verification of client $\mathsf{C}_i$'s input. Therefore, EIFFeL is compatible to our system model (a single server) and the threat model discussed in Section \ref{sec:problem_statement}.

\begin{figure*}[!t]
    \centering
    \includegraphics[width=\linewidth]{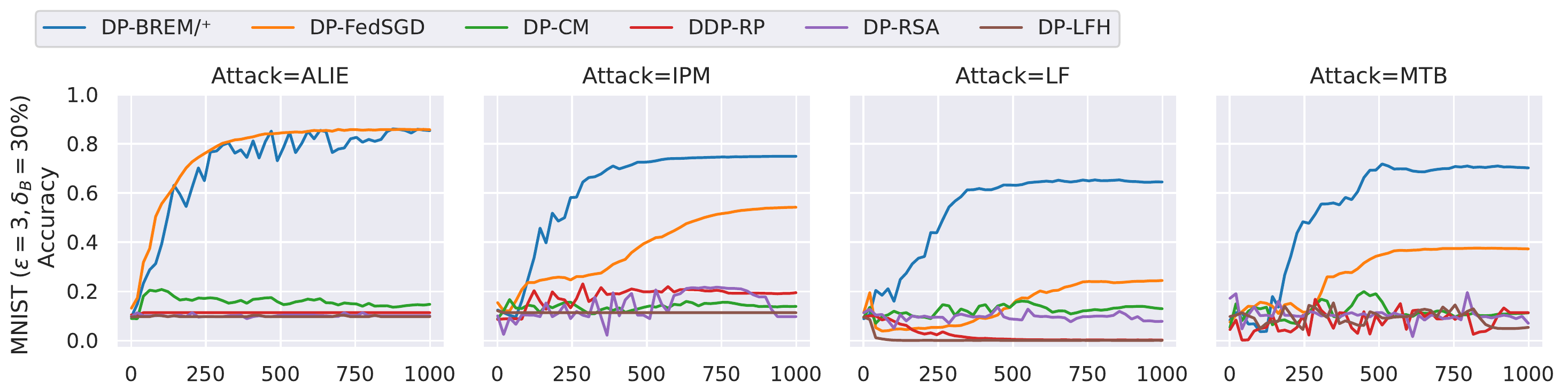}
    \includegraphics[width=\linewidth]{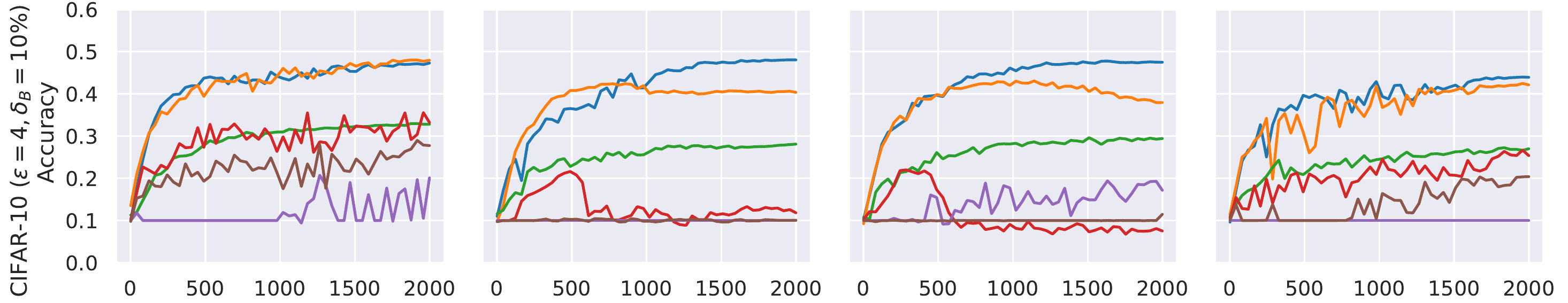}
    \includegraphics[width=\linewidth]{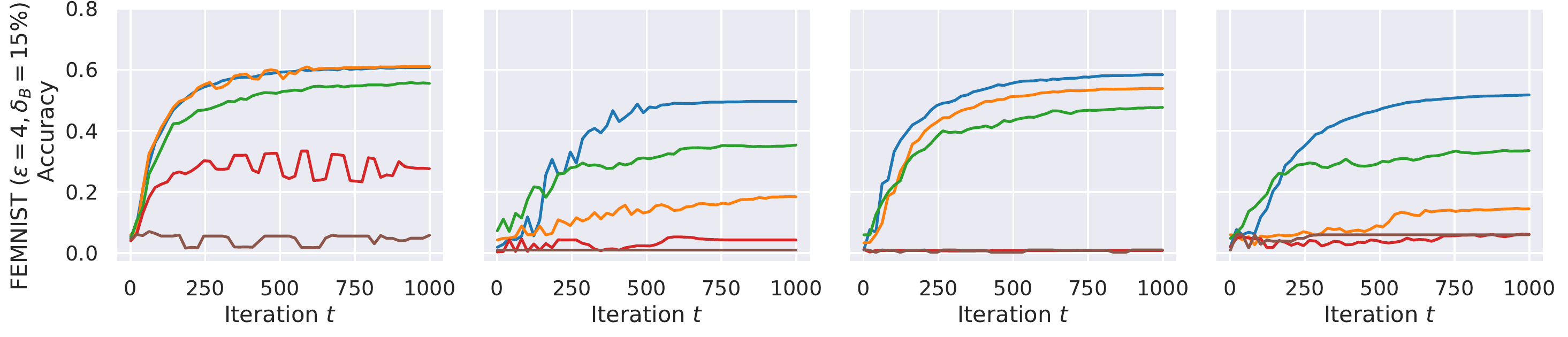}
    \vspace{-8mm}
    \caption{Iteration Curve for three datasets.}
    \vspace{-2mm}
    \label{fig:iteration_curve}
\end{figure*}

\section{Proof of Theorem \ref{thm:security_analysis} (Security Analysis)}
\label{apx:proof_of_thm_security_analysis}

\textbf{Integrity.} We prove that DP-BREM\textsuperscript{+} satisfies the integrity constraint using the following lemmas, where Lemma \ref{lem:integrity_of_input} and Lemma \ref{lem:integrity_of_aggregate} are derived from EIFFeL \cite{roy2022eiffel}.

\begin{lemma}[Integrity of Input]
\label{lem:integrity_of_input}
DP-BREM\textsuperscript{+} rejects all malformed inputs with probability $1-\text{negl}(\kappa)$.
\end{lemma}

\begin{lemma}[Integrity of Gaussian Noise]
\label{lem:integrity_of_Gaussian_noise}
In Phase 2 of DP-BREM\textsuperscript{+}, each client holds the share of random vector $\bm{\xi}$ that follows the Gaussian distribution $\mathcal{N}(0,R^2\sigma^2\mathbf{I}_d)$.
\end{lemma}
\begin{proof}
In the step \textcircled{\small 4} of Phase 2, the jointly generated random number $u$ follows uniform distribution in range $[0,1]$ as long as there is at least one honest client because $u$'s binary representation $\bm{b}$ is the result of bitwise XOR of clients' local random vectors $\{\bm{b}_i\}_{\forall i}$. In step \textcircled{\small 5}, since the utilized MPC protocol \cite{keller2020mp,lindell2017framework} guarantees computation integrity (meaning that the output is correctly computed) with malicious minority, the uniform distribution generated in step \textcircled{\small 4} will be correctly transformed to Gaussian distribution.
\end{proof}

Since clients locally add shares of valid inputs and noise together, DP-BREM\textsuperscript{+} satisfies integrity of aggregate shown in Lemma \ref{lem:integrity_of_aggregate}. Our integrity guarantee in Lemma \ref{lem:integrity_of_aggregate} directly follows  EIFFeL \cite{roy2022eiffel}, though the integrity of noise has different definition compared with the integrity of input. Note that the integrity of EIFFeL (and ours) relies on robust reconstruction property of Shamir's secret sharing \cite{shamir1979share}, and the details can be found from the paper\cite{roy2022eiffel}.

\begin{lemma}[Integrity of Aggregate]
\label{lem:integrity_of_aggregate}
    The aggregated output of DP-BREM\textsuperscript{+} must contain the inputs of all honest clients and the generated Gaussian noise.
    \begin{align*}
        \text{Aggregate}=\sum\nolimits_{i\in\mathcal{I}_H} \bm{z}_i + \sum\nolimits_{i\in\mathcal{I}_M^*} \bm{z}_i + \bm{\xi}
    \end{align*}
    where random vector $\xi\sim\mathcal{N}(0,R^2\sigma^2\mathbf{I}_d)$, $\mathcal{I}_H$ is the set of all honest clients, $\mathcal{I}_M^*$ is the set of malicious clients with well-formed inputs (i.e., $\mathcal{I}_M^*=\mathcal{I}_\mathsf{Valid}\backslash\mathcal{I}_H$)
\end{lemma}

\textbf{Privacy.}
DP-BREM\textsuperscript{+} guarantees: nothing can be learned about a private input $\bm{z}_i$ for an honest client $\mathsf{C}_i$, except:

1) $\bm{z}_i$ passes the integrity check, i.e., $\mathsf{Valid}(\bm{z}_i)=1$ .

2) anything that can be learned from the noisy aggregation of well-formed inputs (thus achieving the same DP guarantee as the original DP-BREM).

We prove this privacy property using the following lemmas, where Lemma \ref{lem:privacy_input_validation} and Lemma \ref{lem:privacy_aggregation} are derived from EIFFeL \cite{roy2022eiffel}.

\begin{lemma}
    \label{lem:privacy_input_validation}
    In Phase 1, for an honest client $\mathsf{C}_i$, DP-BREM\textsuperscript{+} reveals nothing about the private input $\bm{z}_i$ except $\mathsf{Valid}(\bm{z}_i)=1$. 
\end{lemma}

\begin{lemma}
    \label{lem:privacy_noise_generation}
    In Phase 2, DP-BREM\textsuperscript{+} reveals nothing about the generated Gaussian noise.
\end{lemma}
\begin{proof}
In step \textcircled{\small 4}, no entity learns the uniformly random number $u$ as long as there is at least one honest client due to the bitwise XOR operation. In step \textcircled{\small 5}, nothing is revealed because the utilized MPC protocol \cite{lindell2017framework} guarantees information theoretic privacy about the input shares during computation for distribution transmission. Note that the step \textcircled{\small 5} only generates the shares hold by clients without outputting the final result.
\end{proof}

\begin{lemma}
    \label{lem:privacy_aggregation}
    In Phase 3, for an honest client $\mathsf{C}_i$, DP-BREM\textsuperscript{+} reveals nothing about the private input $\bm{z}_i$ except whatever can be leaned from the noisy aggregate. 
\end{lemma}

\section{More Experimental Results}
\label{apx:more_experimental_results}
\textbf{Iteration Curve.}
Figure \ref{fig:iteration_curve} shows how the accuracy changes with the training iterations across three datasets. The CIFAR-10 dataset, being more complex, has a larger total training iterations ($T=2000$), while the other two datasets have $T=1000$. The presence of Byzantine attacks results in less smooth iteration curves compared to the attack-free case.

\section{Other Related Work}
\label{apx:other_related_work}

\textbf{FL with DP.}
Differential Privacy (DP)  was originally designed for the centralized scenario where a trusted database server, which has direct access to all client's data in the clear, wishes to answer queries or publish statistics in a privacy-preserving manner by randomizing query results. In FL, McMahan et al. \cite{mcmahan2018learning} proposed DP-FedSGD and DP-FedAvg, which provide client-level privacy with a trusted server. Geyer et al. \cite{geyer2017differentially} uses an algorithm similar to DP-FedSGD for the architecture search problem, and the privacy guarantee acts on client-level and the trusted server too. Li et al. \cite{li2020differentially} studies online transfer learning and introduces a notion called task global privacy that works on record-level. However, the online setting assumes the client only interacts with the server once and does not extend to the federated setting. Zheng et al. \cite{zheng2021federated} introduced two privacy notions, that describe privacy guarantee against an individual malicious client and against a group of malicious clients (but not against the server) on record-level privacy, based on a new privacy notion called $f$-differential privacy. Note that, our solutions achieve record-level DP under either a trusted server or a malicious server.

\textbf{Byzantine-Robust FL.}
Recently, there have been extensive works on Byzantine-robust federated/distributed learning with a trustworthy server, and most of them play with median statistics of gradient contributions. Blanchard et al. \cite{blanchard2017machine} proposed Krum which uses the Euclidean distance to determine which gradient contributions should be removed. Yin et al. \cite{yin2018byzantine} proposed two robust distributed gradient descent algorithms, one based on the coordinate-wise median, and the other on the coordinate-wise trimmed mean. Mhamdi et al. \cite{mhamdi2018hidden} proposed a meta-aggregation rule called Bulyan, a two-step meta-aggregation algorithm based on the Krum and trimmed median, which filters malicious updates followed by computing the trimmed median of the remaining updates.

\textbf{Private (non-DP) and Byzantine-Robust FL.}
Recently, some works tried to simultaneously achieve both privacy and robustness of FL. He et al. \cite{he2020secure} proposed a Byzantine-resilient and privacy-preserving solution, which makes distance-based robust aggregation rules (such as Krum \cite{blanchard2017machine}) compatible with secure aggregation via MPC and secret sharing. So et al. \cite{so2020byzantine} developed a similar scheme based on Krum, but rely on different cryptographic techniques, such as verifiable Shamir's secret sharing and Reed-Solomon code. Velicheti et al. \cite{velicheti2021secure} achieved both privacy and Byzantine robustness via incorporating secure averaging among randomly clustered clients before filtering malicious updates through robust aggregation. However, these works only ensure the security of the aggregation step and do not achieve DP for the aggregated model. 

\end{document}